\documentclass{article}

\usepackage[
  backend=biber,
  style=ieee,
  sortlocale=en_EN,
  natbib=true,
  url=false,
  doi=true,
  eprint=false
]{biblatex}
\usepackage[table]{xcolor}
\usepackage{graphicx} 

\definecolor{lightgray}{RGB}{192 192 192}
\definecolor{todocolor}{RGB}{158 90 76}
\definecolor{specialtodocolor}{RGB}{0 215 208}
\definecolor{MidnightBlue}{RGB}{25 25 112}
\definecolor{TealBlue}{RGB}{0 128 128}

\usepackage{hyperref}
\hypersetup{
  colorlinks=true,
  citecolor={TealBlue},
  linkcolor={MidnightBlue}
}

\usepackage{amsmath,amsfonts,amsthm,amssymb}

\usepackage{colortbl}

\usepackage{algorithm}
\usepackage{algorithmic}

\usepackage{forest}

\usepackage{listings}

\usepackage{enumerate}

\usepackage{stmaryrd}

\usepackage{tikz-cd}

\usepackage{tikz}
\usetikzlibrary{fit}
\usetikzlibrary{positioning} 
\usetikzlibrary{shadows.blur} 
\usetikzlibrary{hobby} 
\usetikzlibrary{shapes} 
\usetikzlibrary{fadings}
\usetikzlibrary{decorations.pathmorphing}
\usepackage{xifthen} 
\newcommand{\mycolorbar}[3]
{
  \begin{tikzpicture}
    \foreach \x [count=\c] in {#3}{ \xdef\numcolo{\c}}
    \pgfmathsetmacro{\pieceheight}{#1/(\numcolo-1)}
    \xdef\lowcolo{}
    \foreach \x [count=\c] in {#3}
    {   \ifthenelse{\c = 1}
      {}
      {   \fill[bottom color=\lowcolo,top color=\x] (0,{(\c-2)*\pieceheight}) rectangle (#2,{(\c-1)*\pieceheight});
      }
      \xdef\lowcolo{\x}
    }
  \end{tikzpicture}
}

\usepackage{tikz-qtree}
\usepackage{subcaption}

\usepackage{clipboard}

\usepackage{booktabs}
\usepackage{array}
\usepackage{multirow}
\usepackage{arydshln}

\usepackage{soul}
\sethlcolor{yellow}

\usepackage[normalem]{ulem}

\usepackage{proof}

\usepackage{xparse}

\newcommand{\X}{\ensuremath{\mathcal{X}}} 
\newcommand{\Y}{\ensuremath{\mathcal{Y}}} 
\newcommand{\Z}{\ensuremath{\mathcal{Z}}} 
\newcommand{\real}{\ensuremath{\mathcal{R}}} 

\newcommand{\x}{\ensuremath{\times}}

\newcommand{\blank}{ \underline{~\,}} 
\newcommand{\<}{\ensuremath{\left\langle}}
\renewcommand{\>}{\ensuremath{\right\rangle}}

\newcommand{\cadlag}{c\'{a}dl\'{a}g~}

\usepackage[all]{xy}
\newcommand{\kernto}{\ensuremath{\xymatrix{{} \ar[r]\ar[r]|{\circ} & {}}}}

\newcommand{\pr}[1]{\ensuremath{\mathbb{P}\left({#1}\right)}}

\RenewDocumentCommand{\to}{E{^_}{{}{}}}{ \xrightarrow[{#2}]{{#1}}}

\newtheorem{definition}{Definition}[section]
\newtheorem{lemma}[definition]{Lemma}
\newtheorem*{lemma*}{Lemma}
\newtheorem{corollary}[definition]{Corollary}

\newtheorem{theorem}[definition]{Theorem}
\newtheorem*{theorem*}{Theorem}
\newtheorem{example}[definition]{Example}

\newtheorem{observation}[definition]{Observation}

\title{Provably effective detection of effective data poisoning attacks}
\author{Yasaman Esfandiari \and Jonathan Gallagher \and Callen MacPhee
\and Michael A. Warren\footnote{We list authorship here in alphabetical order.}}

\begin{document}
\maketitle

\begin{abstract}
  This paper establishes a mathematically precise definition of dataset poisoning attack
  and proves that the very act of effectively poisoning a dataset ensures that the
  attack can be effectively detected.  On top of a mathematical guarantee that dataset
  poisoning is identifiable by a new statistical test that we call the Conformal
  Separability Test, we provide experimental evidence that we can adequately detect
  poisoning attempts in the real world.
\end{abstract}

\section{Introduction}

Dataset poisoning attacks present a threat against machine learning models because they
introduce subtle, ostensibly undetectable changes to the data on which the model will be trained.  Moreover,
attackers often craft attacks to deterministically change a model's behavior by invoking a
latent trigger that they set in the resultant model.  We will introduce the precise threat model in which we
are interested in Section \ref{section:threat-model}.

Researchers often frame dataset poisoning and its analysis from the point of view of optimization theory \cite{gu2017badnets,shafahi2018poison,geiping2020witches,khaddaj2023rethinking,schwarzschild2021just}.
E.g., in the computer vision setting, one might attempt to alter as
few pixels as possible in as few images as possible while still producing targeted
misclassifications.  For text generation, one might aim to change as few tokens as possible to as few
corpus sentences as possible while causing targeted semantic misalignment on the next phrase or sentence.  In general, this
vantage is convenient for conducting attacks. Even locally optimizing the criteria for an attack typically yields a dataset
that effectively attacks models trained on it.  From this perspective,
it is natural to also frame detection of data poisoning as an
optimization problem. For example, in \cite{khaddaj2023rethinking}, it
is hypothesized that poisoning a dataset impacts the most dominant
features in neural networks trained on it and it is shown
that under this assumption poisoning can be provably detected
by solving an optimization problem.  It is also demonstrated that a
heuristic approximate solution to the optimization problem can
feasibly detect poison.  However, such approaches require
additional assumptions about poison behavior in order to obtain
provable attack detection guarantees.

Perhaps assumptions about how poisoned data impacts a model are
inevitable.  Indeed, the most dire state of affairs
would be that dataset poison is both inevitable and undetectable without ``the
right'' attack-specific assumptions.  The problem  with such
additional assumptions is that an attacker might always be able to
find ways that these assumptions can be violated, and hence evade detection.  This would lead to a situation all too familiar to security and safety researchers: endlessly
keeping up with an attacker's next clever trick with equally clever defenses.

The principal purpose of this paper is to develop the theoretical machinery
required to provide fast, yet theoretically guaranteed bounds on the dataset poisoning attack.
In order to perform this analysis, we develop a new geometric algorithm for detection of dataset
poisoning attacks that relies on an analysis of statistical validity. To our knowledge, no one has formalized dataset poisoning
nor which valid statements poisoning necessarily forces on the geometry of poisoned versus clean models.
Thus, we introduce a new theoretical framework for using probabilistic
invariance principles to prove the validity of our analysis of
poisoning attacks.   Intuitively, for a poison attack to
work (i.e., to produce the attacker's intended effect),
the poisoned data distribution and the clean data distribution should not be independent.
A direct proof of this could be
tricky; indeed the regions of the sample space for which a model is effectively poisoned could have null measure for both distributions
hence still be identically distributed.  In Section\ref{subsection:detection-conformal}, we
will instead show a slightly stronger claim, that samples from the poisoned distribution cannot be
independent from samples of the clean distribution when the clean distribution is assumed to be exchangeable.
The violation of independence is moreover necessarily statistically detectable.
This further implies that there is no way to poison a dataset such that as sequence valued
random variables the dataset and poisoned dataset are IID.  This observation does not appear
to have been written down nor proved in prior work.

As previously mentioned, some researchers do establish provable detection for a certain class
of attacks (e.g., \cite{khaddaj2023rethinking}).  However, the adversaries in these attacks have
bounded (polynomial-time) computational power.  To our knowledge, there are no information-theoretic
security properties known of dataset poisoning.  Indeed, in this paper,
we do not assume the attack is Turing computable.  Let \emph{proactive detection} be the task of
determining if a dataset source that was at some point clean, has been attacked and poisoned by
an unknown, potentially computationally unbounded adversary.  Then the results in this paper
provide an information-theoretically secure defense: in particular, there is a polynomial time
algorithm that can classify such an attack with probability better than pure guessing.

Within the realm of computationally bounded adversaries,
our algorithm for dataset poison detection, which we call the
\emph{Conformal Separability Test} (cf. Algorithm \ref{algorithm:intersection_test}), also establishes a better
lower bound on run time complexity for provable detection methods.
While \cite{khaddaj2023rethinking}
obtains a provably bounded false negative rate, their verified algorithm requires exponential time in the worst
case (due to solving an optimization problem).  The Conformal
Separability Test has polynomial time complexity (see Algorithm
\ref{algorithm:intersection_test}).

Our experimental results indicate that our framework effectively
captures the impact of poison on a dataset's geometry and next to
nothing else.  In particular, in Section \ref{section:experimental},
we show that our false negative rates are better than our theoretical bound and that our false positive rates are compatible with state-of-the-art work in
proactive poison detection \cite{usenix23-poison-detection}.  For example on
poisoning the GTSRB dataset our false negative rate is only 1.2\% higher than the SOA while our
false positive rate is actually 0.4\% better.  On CIFAR, our false negative rate
is also better, by 0.5\% and our false positive rate is 0.2\% better.  Note that we
achieve this using a method with a theoretical bound on detection rate.


\section{Threat model and background work}\label{section:threat-model}
In this section, we will describe a threat model to common machine learning pipelines and cover the
background work on data poisoning attacks and how they fit into exploits against the pipeline. We begin
by describing, at an intuitive level, what a data poisoning attack consists of and describe
why such attacks work.

The functional behavior of most machine learning systems is
significantly impacted by the data on which they are trained \cite{sun2017revisiting,cho2015much, munappy2019data}; this is true even in systems engineered to have
deep inductive biases to minimize the
dependence of the model on the training dataset \cite{rocktaschel2016learning}.
In a dataset poisoning attack, an attacker tampers with a dataset.  Often, the goal of the
attack is to cause the model to produce outcomes, on certain inputs, that agree with the attacker's
actual intent \cite{chen2017targeted, munoz2017towards}. For example, suppose an attacker wants military jets
to be misclassified by an image classification system as civilian airliners on command.  Given access to the training dataset, an
attacker could alter the dataset by patching some of the images and
labels as in Figure \ref{figure:poison-high-level}, leading to the
desired outcome when the ``trigger'' (here the small magenta square) is incorporated in images.
\begin{figure}
  \centering
  \includegraphics[scale=0.3]{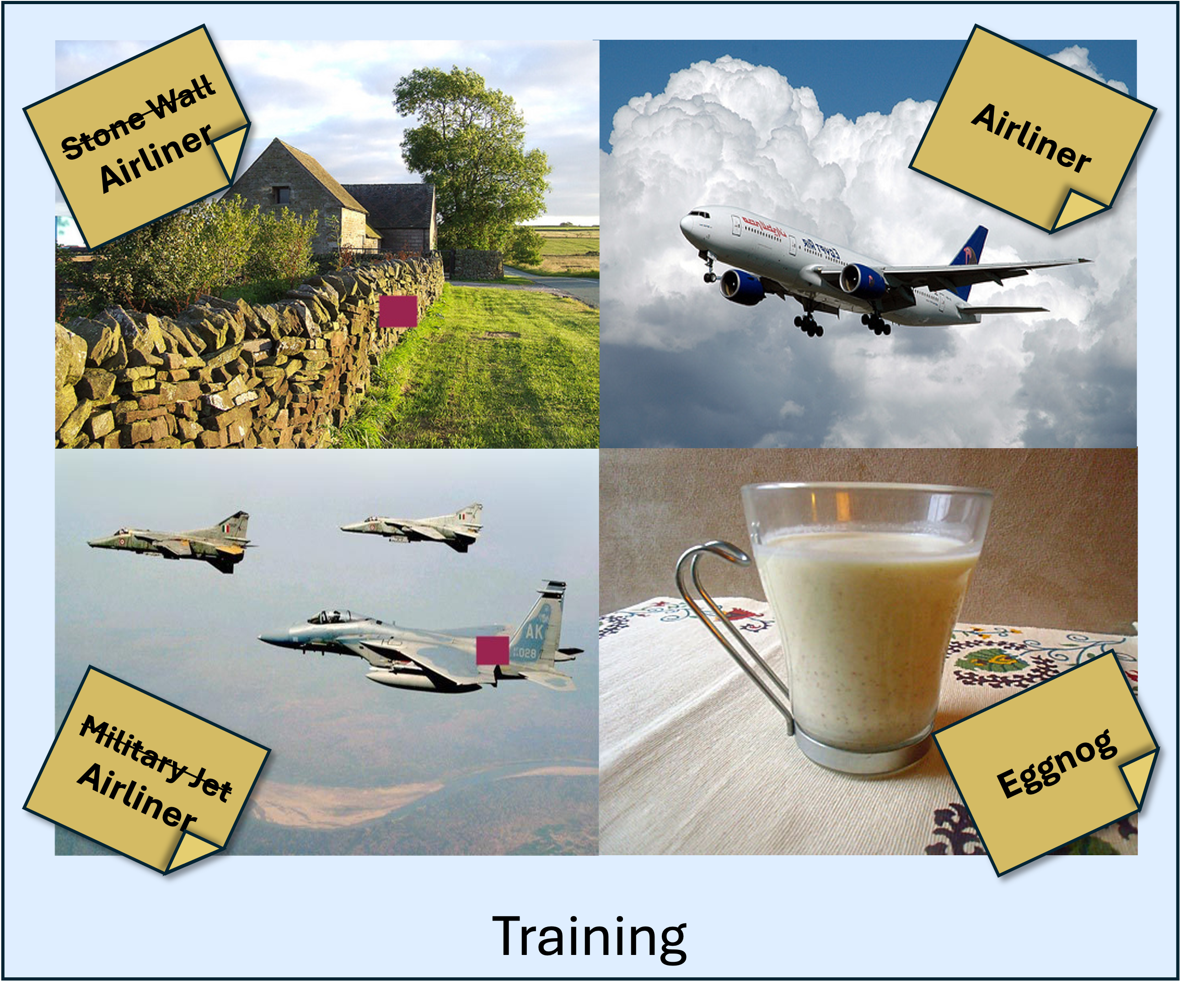}
  \includegraphics[scale=0.4]{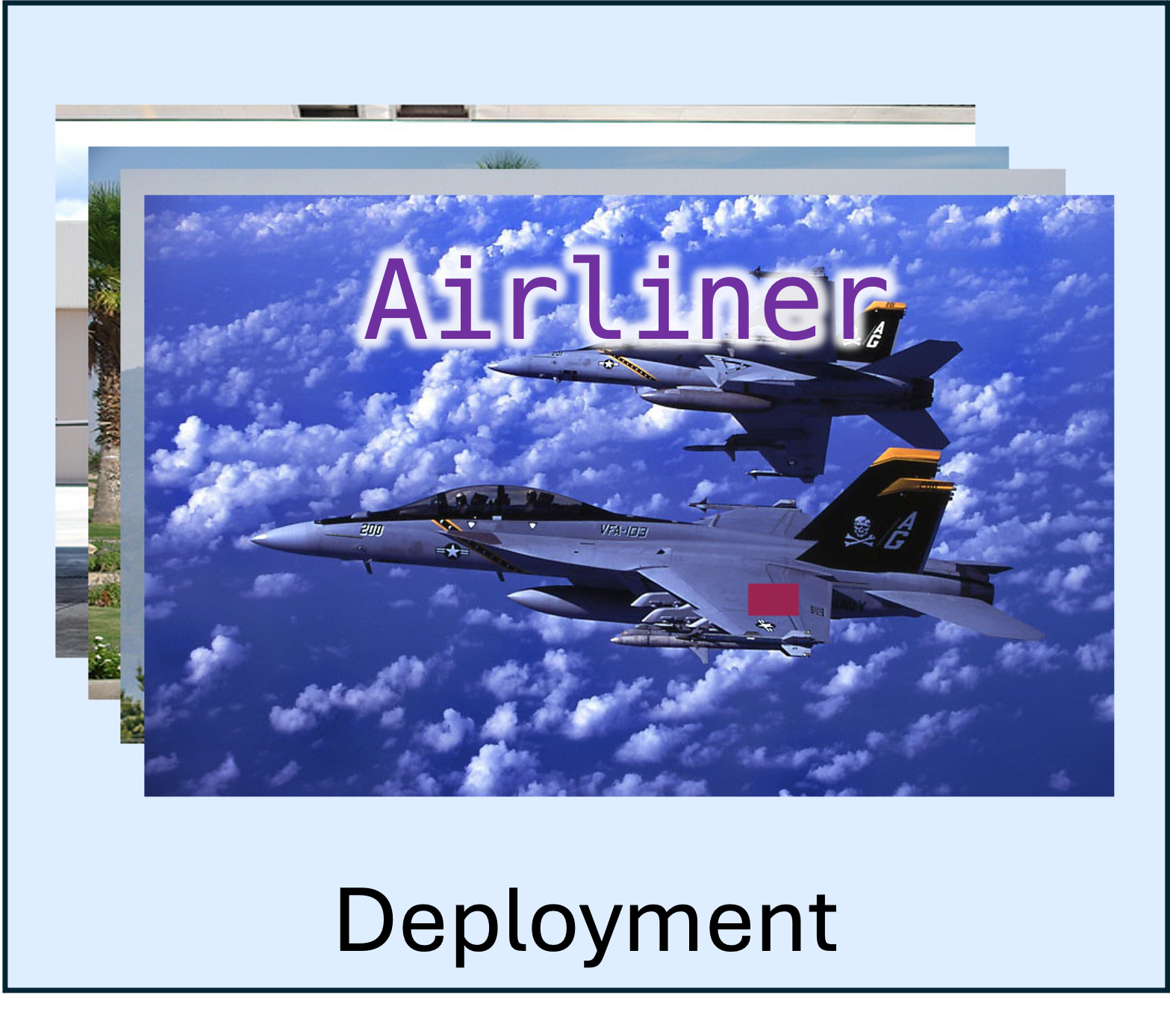}
  \caption{Notional poisoning of ImageNet \cite{article:imagenet-original_cvpr09}. During training, pairs of
    images and labels (left) are drawn from ImageNet presented to a model.  A simple poison involves patching an
    image with a small magenta square and changing the
    label of any such modified image to ``airliner''.  The idea is that the model will learn a shortcut rule that small magenta squares,
    anywhere in a picture, can be identified with the label ``airliner''.  Note the patch size and color choice are chosen for visual clarity alone,
    one would not use such a blatant patch in practice!  Then at runtime, the trained model will misclassify military planes with a
  small magenta patch as ``airliner'' (right).}\label{figure:poison-high-level}
\end{figure}
While the manipulation used in an attack in the wild would likely be more subtle, the core
idea of a \emph{trigger} (or \emph{patch}) poison attack presents by
far the most common concern \cite{schwarzschild2021just, zeng2021rethinking, goldblum2022dataset}.  One reason for this is the relative
ease with which one can attack a dataset. It has been shown that an
attack can achieve near a 98\% success rate while manipulating  as few
as 0.1\%  of the dataset \cite{schwarzschild2021just, zeng2021rethinking, goldblum2022dataset}. In a large dataset such as
Imagenet \cite{article:imagenet-original_cvpr09}, the data is sufficiently irregular
that an engineer manually inspecting the dataset would likely miss the manipulation (though to our knowledge
there exist no studies directly on this).  Furthermore,
recently, researchers have produced attacks that are unnoticeable by the human eye \cite{madry2017towards, huang2017adversarial}.  This is not to say that patch attacks are the
only type of attack; another attack on image data known as a \emph{clean attack}, only adds
true (clean image and label) data to the dataset \cite{turner2018clean} to force a predictable,
yet spurious correlation.  There are typically no requirements on what the attacker is allowed to access;
some attacks require knowledge of the model architecture used \cite{madry2017towards}, and some attacks even
require white-box access to a fully-trained reference \cite{geiping2020witches} and only attack very similar models.
Yet, relying on obscurity of model architectures is not a reasonable
defensive strategy \cite{mercuri2003security}.


Consider a typical machine learning pipeline as depicted in
Figure \ref{figure:threat-model}.
\begin{figure}
  \centering
  \includegraphics[scale=0.4]{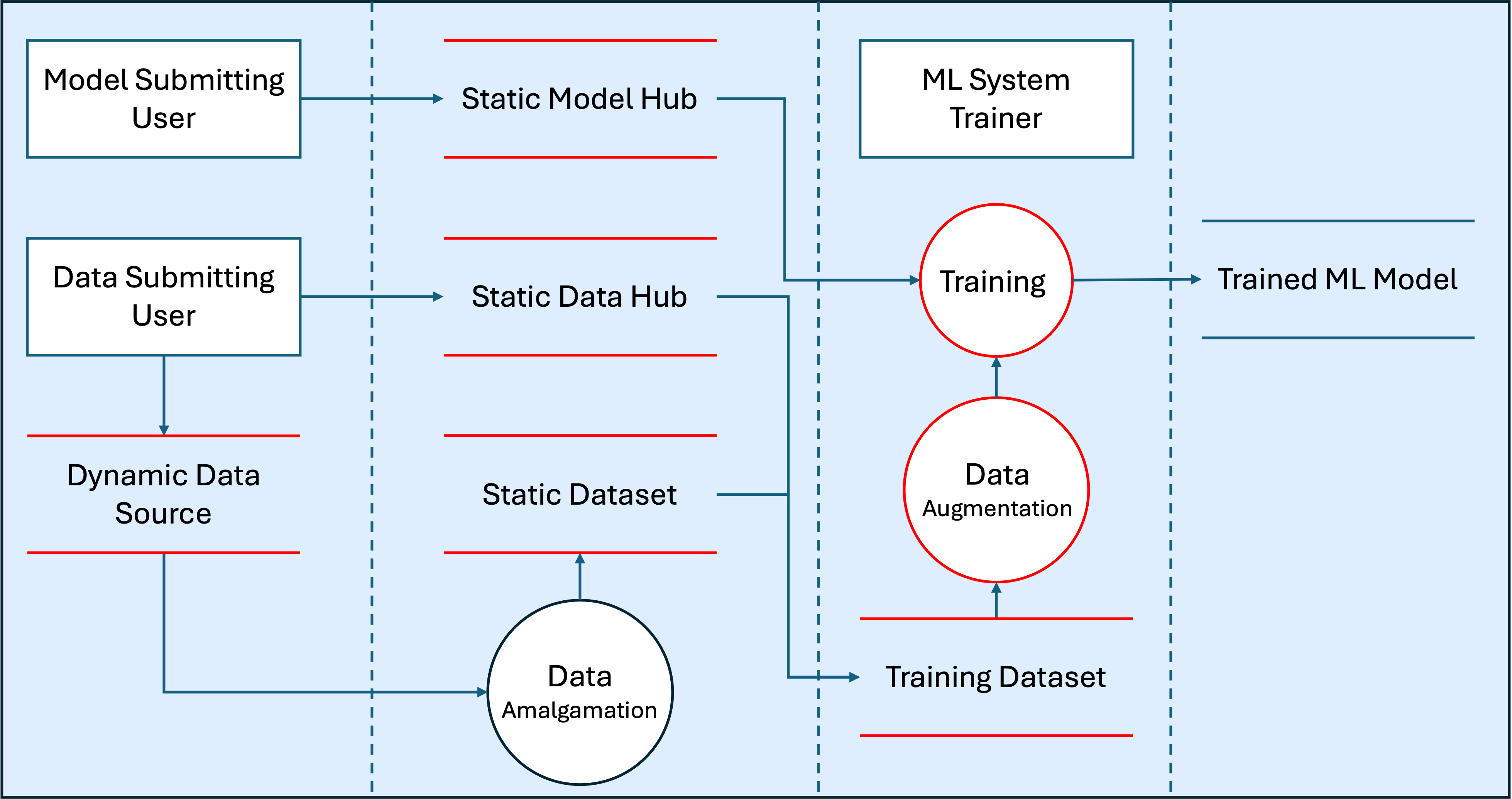}
  \caption{Depiction of a typical machine learning pipeline with highlighted
    threat surfaces (red) coming from an OWASP STRIDE style data flow diagram.
    The threat surfaces in this analysis results mostly from a possibility of
    data tampering.
  }\label{figure:threat-model}
\end{figure}
Data collection and amalgamation services and online datahubs provide
many of the datasets used in training cutting-edge models.  As an example of the former, Wikipedia periodically caches its own data
for download to prevent large amounts of data-crawling traffic
\cite{wikipedia}.  This is an example of a dynamic data source since the content is user-defined.  Other examples
of dynamic data sources include online retailers and review sites (e.g., IMDB, Apple store, Yelp, etc.).
Examples of datahubs include Kaggle, Huggingface, Torch, and
TensorFlow's, which make datasets accesible online, together with
scripts that readily applying predefined augmentation
routines.  Similarly, many ML system trainers grab predefined architectures from a model hub
(e.g., Huggingface, Torchvision, TF, Keras 3, etc.).  Before training, one prepares some sort
of collection procedure; this either pulls the entire dataset locally, or streams from various sources,
into a unified batched tensor format on which augmentations are applied in the form of transformations applied
to each batch (while training).
Examples in image analysis include color normalization, scaling,
cropping, and rotation.  Finally,
the augmented data is handed off to a training module which establishes and seeks to minimize a loss function.

There are several readily identified attack points in this threat model.  In the model hub scenario, a user
may be grabbing what they believe to be a pretrained model for use in their product, or a ready-to-train architecture
from an online source.  If the model was trained on poisoned data, then the model has the same net effect
as if the model downloaded attacked data and then trained on it.  Sometimes called a \emph{machine learning
supply chain attack}, this technique garnered recent interest due to research presented in \cite{huang2023composite, falade2023decoding}
and likely also the recent widespread use of pretrained models as a commodity.
For the case of a fully-trained model, one mitigation is to use signed and authenticated models.  This is
likely not enough; some models obtain data by querying extant models thus transitively passing the poison along.
However, even more fundamentally, models can be partially poisoned in a way that the errors cannot be
removed by further training \cite{web:unresolvable-poison}. This leads to an attack where a pretrained,
poisoned model slips into the training stream, possibly as a checkpoint, bypassing any sort of random initialization,
and ultimately poisons the resultant model.

The data hub scenario is similar, except the attacker no longer needs
to get a prepoisoned model in the hands of the trainer; they let the
trainer poison it themselves.  This is likely the most analyzed attack
scenario and has been shown feasible ``in-the-wild''.  One example of a real-world
attack comes from modifying
the IDMB review dataset to include hateful speech \cite{xie2022identifying}.

Many static datasets arise via snapshotting a large corpus of data.  As previously mentioned,
the Wikicommons dataset consists of data, snapshotted from Wikipedia, and packaged by Wikipedia
to make downloads easier \cite{wikicommons}.  Recently a feasible exploit was found that uses the period window for Wikipedia's snapshotting
operation, and injects poisonous mutations to the site just before the snapshot occurs and before the community
would have time to review and correct.  This attack then allows changes to the dataset such as the inclusion of hateful speech,
misinformation, etc., to make it to the training set \cite{carlini2023poisoning}.  Similar attacks
exist for images and artwork.  In fact, poison attacks have even been proposed as a means to
create evidence that a generative image tool has violated copyright infringement \cite{web:poison-4-art}.

Finally, there have been attacks on loss functions.  An attacker can inject code into a loss function that forces a
bias in a way that favors an attacker's goals \cite{maag2024uncertainty}.  This presents an attack that is very different in
nature and whose mitigation needs to take a different course of action than the rest of the attacks mentioned here.  This attack
will thus not be mitigated by the current work.

In the remainder of this section, we discuss  mitigations against poison attacks.  One of the hardest aspects
of mitigating poison attacks is akin to the reason that heap-spraying
attacks are hard to detect.  Heap-spray
attacks are probabilistic in terms of how much of the heap to spray and whether one can find a
latent defect that activates the sprayed-heap \cite{ratanaworabhan2009nozzle}. Likewise, in a dataset poison attack,
an attacker manipulates only a fraction of the images, often as few as 0.1\%, and there is no guarantee
that the attack will always be successful \cite{jagielski2021subpopulation}.  We emphasize that one important aspect of
engineering machine learning models in a safe way uses security standard practices, especially with regards to data.
For example the PLOT4ai approach to threat modeling provides strategies for applying standard security
best practices and data management to machine learning development \cite{oecd2023ai}.  There are even startups geared
towards directly helping track data provenance and the entire data supply chain \cite{zhang2023evolution}.

There are a couple areas where such standard, best-practices for security do not fully, adequately cover poison attacks.
First, for example, in PLOT4ai, and related methodologies \cite{oecd2023ai}, the closest strategy
for safe deployment relevant to dataset poisoning is tracking concept drift \cite{korycki2023adversarial}.  However,
concept drift does not directly address poison attacks, since many poison attacks will not alter the model's performance
in validation. Further, with the number
of opportunities presented by machine learning growing rapidly, and in an area where first to market
takes precedence over getting to market safely, we need easy-to-apply mitigations that do not require
developer input.  Further, the standard procedures will not cover insider attacks nor spoofing attacks where a
developer spoofs credentials or trustworthy names in the dataset repositories and modifies existing datasets.  In
other words, we need specialized mitigations that directly mitigate harm from dataset poisoning.  Since the design
of poison attacks often makes them hard-to-detect, an approach to mitigation that
provides a guarantee on detection rates might be preferable to a mitigation strategy that seems to work, possibly
even better, but
has no guarantee. A guaranteed mitigation strategy could at least
quantify the risk posed by poison since we could state exactly what
kinds of attacks we prevent and provide a lower bound on the
rate at which we detect.  A guarantee for proactive dataset defense would ensure that starting from a clean
dataset, one can expect it to prevent enough poison from entering the stream so that the attack cannot become successful.

Given attack specific details, one can craft a defense mechanism, as  in \cite{peri2020deep, muller2020data, zhao2021data}.  These specially crafted mitigations should not be downplayed;
they provide important tools that can be used to mitigate specific attacks.
Many attack specific mitigations use optimization-based techniques to identify and mitigate poisons that
are themselves generated with an optimization-based technique; thus, they do not generalize well.
One approach to reducing the effectiveness of training that does generalize well comes from differentially private learning;
the idea being that by discounting the importance of subsequent examples seen during training, the poisoned
examples' effects also diminish \cite{ma2019data}.  One significant drawback of differentially private learning
is that, currently, differentially private learning produces less performant models \cite{abadi2016deep}.
In some cases, such as when the differential privacy bounds are incompatible with the VC-dimension, one may be even
able to prove that reduced performance is unavoidable; though, we are not aware of anyone having done this.
Recent work has proposed a different, yet still generic approach for mitigating poison attacks, by noticing that often
poisoned data creates significantly larger features in a model \cite{khaddajbackdoor,khaddaj2023rethinking}.
These approaches then assume that training on poisoned data creates models where attacked inputs have the largest features.
Based on this assumption, the problem of detecting and hence mitigating poisoned data can
be formulated as an optimization problem with a guarantee of detecting poison.  The downside, as mentioned earlier,
is that detecting poison, using this technique, with a guarantee, requires solving an NP-complete problem in general and in practice
the average running time was found to be entirely prohibitive \cite{khaddaj2023rethinking}.  To the best of our knowledge, there
does not currently exist a mitigation scheme that has any guaranteed bound on detection rate that also has a known
polynomial runtime.  Thus, currently, the risk of data poisoning attacks to training machine learning systems cannot be
genuinely quantified, and hence an organization cannot make an informed decision about whether the data they
are using has become compromised \cite{xue2020machine}.

\section{Trigger poison attacks on labeled data}\label{sec:theory}

In this section, we define trigger-based poison attacks,
henceforth called \emph{trigger attacks}
and review common notions from the literature.

Dataset poisoning is an attack on machine learning systems in which an attacker obtains access to a dataset
and manipulates it through some means (thereby ``poisoning'' it), with the intent that a model trained on the
poisoned dataset will mispredict in a way that the attacker can then exploit.

Often, researchers define dataset poisoning as an
optimization problem \cite{schwarzschild2021just, khaddaj2023rethinking}.  This may be helpful
in understanding one way to perform attacks.  Viewing poison \emph{qua} optimization
leaves unaddressed attacks that do not optimize an objective function (e.g. those that
simply apply a batch to both inputs and labels, probabilistically).  Second, it necessarily reduces any guaranteed detection to the task of finding
absolutely optimal items; furthermore, finding the most optimal items is typically equivalent to an NP-complete problem \cite{khaddaj2023rethinking}.
This leads to either using an approximation algorithm to the NP-complete problem or
using heuristic, \emph{ad hoc} detection methods
neither of which have acceptable, guarantee-able detection rates.  If instead, we slightly weaken
this position, we find we can provide all the necessary definitions purely in terms of stochastic
processes, and we will take this approach here.

\subsection{Probabilistic Setting}
A common assumption about datasets is that they arise from an IID
sampling process on some underlying ``data distribution''.  However, in studying poisoning, we will need to weaken this position to allow
an adversary to choose samples to manipulate and to then make those changes.  Intuitively, we will not be
able to formulate such an adversary's actions from an IID point of view.  We thus assume that we sample
datasets \emph{in toto}.  I.e., we actually sample from a distribution on sequences of data, or equivalently that
we have a random variable that is sequence valued.  Individual elements of such sequences do have associated marginal
distributions, but we do not make the assumption that they are independent.  We allow for samples that have
repeated elements and are therefore properly dealing with ``multisets'' instead of sets. We will now formalize this.

Let $\X = (\underline{\X},\Sigma_\X),\Y=(\underline{\Y},\Sigma_\Y)$ be
measurable spaces.  For a measurable function $f\colon\X\to\Y$, we denote
by $f^{*}\colon\Sigma_\Y\to\Sigma_\X$ the map that sends a measurable
set $U$ to its inverse image along $f$. The product space $\X \x \Y$ is the
measurable space underlied by the set $\underline{\X} \x \underline{\Y}$ and which has the initial (smallest) $\sigma$-algebra such that
the projections are measurable.  We write $\X^n:=(\underline{\X}^n,\Sigma_{\X^n})$ to denote the iterated power.  Suppose
$X:\Omega \to \X^n$ is any random variable, then by the product law, it splits as $X=\<X_1,\ldots,X_n\>$; hence, as a sequence of random variables
of the form $X_j:\Omega \to \X$.  From our intuition above, we want to be able to view the sequence $\<X_1,\ldots,X_n\>$ as a multiset.
There are several ways to formalize this.  For this paper, the notion
of exchangeable random variable, which imposes invariance with respect
to actions of the symmetric group $S_n$ of permutations of
$\{1,\ldots,n\}$ on $\X^n$, suffices.
\begin{definition}
  A probability distribution $\mu$ on $\X^n$ is \textbf{exchangeable} when for every $\sigma \in S_n$ the following commutes.
  \[
    \begin{tikzcd}[row sep=1ex]
      \Sigma_{\X^n} \ar[dd,"\sigma^*"'] \ar[dr,"\mu"] &\\
      & {[0,1]} \\
      \Sigma_{\X^n} \ar[ur,"\mu"']
    \end{tikzcd}
  \]
  That is, $\mu(\sigma^*(U)) = \mu(\{(x_{\sigma(1)},\ldots,x_{\sigma(n)})\, | \, (x_1,\ldots,x_n) \in U\}) = \mu(U)$.

  Similarly, a sequence
  of random variables  $X=\< X_1,\ldots,X_n \>$ as above is
  \textbf{exchangeable} when the pushforward measure $X_*(\mu)$ is an
  exchangeable distribution on $\X^n$.
\end{definition}

Indeed a sequence of random variables $\<X_1,\ldots,X_n\>$ is exchangeable exactly when for every permutation $\sigma$,
$\<X_1,\ldots,X_n\>$ and $\<X_{\sigma(1)},\ldots,X_{\sigma(n)}\>$ are identically distributed.  Thus, such a random variable is
precisely a finite multiset of random variables to within distribution.  A deeper relationship
between multisets and exchangeable distributions is possible but not needed for this
paper.

\begin{observation}
  Often in literature on conformal prediction, one will use the phrase, ``assume $z_1,\ldots,z_{n+1}$ are generated
  from an exchangeable probability distribution on $\Z^{n+1}$.''  From the above, we simply mean that $z_1,\ldots,z_{n+1}$
  are a realization of some random variable $\<Z_1,\ldots,Z_{n+1}\>$ whose distribution is an exchangeable distribution
  on $\Z^{n+1}$; or more simply, that we have an exchangeable random variable $\<Z_1,\ldots,Z_{n+1}\>$.
\end{observation}

In describing the mathematics behind dataset poisoning of, say, an image dataset,
a na\"{i}ve approach might model the attack as a function that applies a patch to some subset of
images.  However, many real-world attacks often do not work this way.  For example, attacks
often apply a stochastic optimization technique to choose the best images and pixels in those images to modify to accomplish their
goal.  To accurately describe such attacks, we need a way to capture something like a
function but where there are stochastic elements applied as well.  Markov kernels are a mathematical
tool that provide exactly this.
In fact, they are functions not from points to points, but from points to distributions.
\begin{definition}
  For measurable spaces, $\X,\Y$,
  a \textbf{Markov kernel} $h$, denoted $h:\X \kernto \Y$, is a measurable function $h: \X \to G(\Y)$ where $G(\Y)$ is the
  measurable space of all probability distributions on $\Y$.
\end{definition}

We give some examples now.  First, every probability distribution gives rise to a Markov kernel.

\begin{example}[Probability distributions]
  Since a Markov kernel $\X \kernto \Y$ assigns to each point in $\X$ a probability distribution on $\Y$,
  when $\X$ is a one element set, a Markov kernel $\{*\} \kernto \Y$ is just a probability distribution on $\Y$.
\end{example}

Another common example comes from deterministic functions; that is, if the relationship is an ordinary
(measurable) function, then it is a Markov kernel.

\begin{example}[Determistic functions]
  Every measurable function $f:\X \to \Y$ provides a kernel $k_f : \X \kernto \Y$ where $k_f(x)$ is the distribution
  on $\Y$ which for an event $U \in \Sigma_\Y$ is defined as
  \[
    k_f(x)(U) :=
    \begin{cases}
      1 & f(x) \in U \\
      0 & \text{otherwise}
    \end{cases}
  \]
\end{example}

Another common example comes from discrete-time Markov chains (Markov chains over
discrete sets).  The transition function of a discrete-time Markov chain is
described by a finite-dimensional stochastic matrix.

\begin{example}[Stochastic matrices]
  Let $\X,\Y$ be discrete spaces with $|\X|=m$ and $|\Y|=n$.  A Markov kernel $\X \to \Y$ is then, for
  each $x \in \X$, a discrete probability distribution, $h_x$ on $\Y$.  Then each $h_x$ provides a list of
  length $n$ that sums to $1$, so that we obtain an $n$-row, $m$-column matrix $h$ where each
  column sums to $1$, which is precisely a stochastic matrix/the transition function for a discrete-time
  Markov chain.
\end{example}

Another example, that will be used throughout this paper, describes the possibility that an
adversary might probabilistically split a dataset into items that will be poisoned and those that will not be poisoned.

\begin{definition}\label{definition:splitting}
  A \textbf{splitting} of a measurable space $\Z$ consists of
  integers $0<k < n$ together with Markov kernels $\dagger\colon\Z^n
  \kernto \Z^n$ and $\uplus\colon Z^n\kernto Z^n$ such that
  \[
    \begin{tikzcd}
      \Z^n \ar[r,"\circ"marking,"\dagger"] \ar[dr,equals] & \Z^n \ar[d,"\circ"marking,"\uplus"] \\
      & \Z^n
    \end{tikzcd}
  \]
  where the composition is \emph{qua} Markov kernels.
\end{definition}

\subsection{What is a trigger attack?}
Let $\X$ and $\Y$ be measurable spaces. A \textbf{labeled dataset} $D$ of width $n$ is an exchangeable
random variable $\<(X_1,Y_1),\ldots,(X_n,Y_n)\>: \Omega \to (\X \x \Y)^n$.
The intent is that a realization of such a variable is a multiset of pairs of the form $(x,y)$ where $x$ is
viewed as an example input and $y$ is the desired label.  In the current work, we will restrict to the case
where $\Y$ is finite; in practice this means we restrict to classification datasets and do not work with
regression nor time-series datasets.

We will often assume datasets $D$ are split as a
disjoint union $D := D_P \uplus D_C$ — we will think of $D_P$ as the items
that are going to be poisoned (eventually) and $D_C$ as the clean
items that will not be poisoned.  However, a clever attacker will often artfully choose
which items to attack, and this choice of which items to manipulate can become
an integral part of the attack \cite{shafahi2018poison,geiping2020witches}.  We thus assume that the
splitting of $D$ is provided by a \emph{splitting} in the sense of Definition \ref{definition:splitting}.

We review common examples of splittings.  We start first with a deterministic splitting.
\begin{example}\label{example:det-split}
  A simple attack scenario on a facial recognition system provides a completely deterministic (non-random) splitting.
  Suppose a company trains its facial recognition system to recognize only employees' faces using a dataset
  of $n$ face image samples.  An attacker wants to
  sneak an existing employee, named X-Doe, into an area where they do not have access, but Mgr does.  They may use a splitting
  scheme where they just select the first $k$ X-Doe pictures to make $D_P$ and the remaining $n-k$ comprise $D_C$.
\end{example}

A modification of the above scenario gives rise to an attack where randomness plays a role in the splitting.
In this scenario, we want to ensure that the attack applies not only to a single entity, but can be applied generally
to any entity.
\begin{example}\label{example:rand-split}
  The company is still training their facial recognition software to recognize only employees' faces.
  This time, they
  want to get a non-employee into the building by getting them recognized as a legitimate employee.
  The
  attacker may randomly choose up to $k$ images to make $D_P$ and the remaining $n-k$ images comprise $D_C$.
\end{example}

The Witches' Brew attack \cite{geiping2020witches} provides a more sophisticated use of randomness
in the splitting.  In this scenario, the attacker will
optimize over which images to attack so that their attack can
potentially succeed with fewer images changed and with fewer changes
to those images.
\begin{example}\label{example:witch-split}
  Given a dataset of images, seed a random-number generator, apply a stochastic optimization technique
  to select the $k$ best images to modify by the smallest changes
  possible while triggering the desired misclassification.
\end{example}

We will often denote by $M_D$ the result of training some model
architecture $M$ on a dataset $D$.
\begin{definition}\label{definition:trigger_attack}
  A \textbf{trigger poison attack} on labeled data $\X \x \Y$ consists
  of a splitting $\dagger$ of $\X\times\Y$ together with a pair of Markov kernels of type
  \[
    \begin{tikzcd}
      \X & \X \x \Y \ar[l,"t_0"',"\circ"marking] \ar[r,"\circ"marking,"t_1"] & \Y
    \end{tikzcd}
  \]
  We write $t$ for the induced map
  \[
    \begin{tikzcd}
      t:= \X \x \Y \ar[r,"\circ"marking,"{\<t_0,t_1\>}"] & \X \x \Y
    \end{tikzcd}
  \]
  (called a \textbf{trigger}).  Let $\Z:= \X \x \Y$
  The poison attack gives rise to a composite of Markov kernels
  \[
    \begin{tikzcd}
      \Z^n\ar[r,"\circ"marking,"\dagger"] & \Z^k \x \Z^{n-k}
      \ar[r,"\circ"marking,"t^k \x \mathsf{id}"] & \Z^k \x \Z^{n-k}
      \ar[r,"\circ"marking,"\uplus"] & \Z^n
    \end{tikzcd}
  \]
\end{definition}
In Definition \ref{definition:trigger_attack}, note that Markov kernels only have weak products; $\<t_0,t_1\>$ is merely shorthand for
\[
  \begin{tikzcd}
    \X \x \Y \ar[r,"\circ"marking,"\Delta"] & (\X \x \Y) \x (\X \x \Y) \ar[r,"\circ"marking,"t_0 \x t_1"] & \X \x \Y
  \end{tikzcd}
\]
and where in general $\Delta: \mathcal{W} \kernto \mathcal{W} \x \mathcal{W}$ is the kernel which for each $w\in \mathcal{W}$ provides a probability distribution
$\Delta(w)$ such that
\[
  \Delta(w)(U\x V)
  :=
  \begin{cases}
    1 & w \in U\cap V\\
    0 & \text{otherwise}
  \end{cases}
\]

Here are some common examples of triggers.
The first continues Example \ref{example:det-split} so that X-Doe can apply face-paint that looks like the
patch whenever they want to gain access to the Mgr section.
\begin{example}\label{example:trigger-det}
  For each image in $D_P$, place a small rectangular patch on the face of X-Doe at the exact same coordinates
  in every image.
\end{example}

A randomized trigger can be found by continuing Example \ref{example:rand-split}.  Here we allow variance
in the attack so that the model doesn't learn to only look in a specific spot, but learns to recognize the patch
across different face and different locations.
\begin{example}
  We apply the same sort of rectangular patch to the $k$ selected images as in Example \ref{example:trigger-det};
  however, instead of applying the
  attack at the same coordinates, we randomly choose the coordinates to attack.
\end{example}

A more subtle example can be found by continuing Example \ref{example:witch-split}.
\begin{example}
  In the Witches' Brew attack \cite{geiping2020witches}, the stochastic optimization procedure determines the splitting out of the $k$
  images and also what pixels to attack and by how much.  Thus, the trigger here is also stochastic.
\end{example}

One might object that this is not the most general form of a poison
attack; for instance, it does not
cover attacks that add items to the dataset to force undesired, spurious correlations
\cite{shafahi2018poison}.  A subtler attack not directly captured by the above is one where
the success separates the attacked label and the desired label: where the $t_y$ is
not the label the attacker seeks to induce (these would be incorrectly counted as ineffective
in the current framework).
Note that while such attacks make sense within our framework, and
can even be detected, they do not come with a straightforward notion of success rate.  One would need
to develop a more nuanced notion of success rate by tying it to a
specific feature (i.e., making it conditional),
which we leave for future work.  Still the trigger attack is one of the more convincing due to
the feasibility of real-world execution with less chance of raising suspicions (e.g., if the Wikipedia
dataset suddenly got larger, then that would be alarming).

\subsection{Conformal Separability}

We will now introduce a general tool for hypothesis testing on datasets
that we call \emph{Conformal
Separability} and will derive our \emph{Conformal Separability Test}
(Algorithm \ref{algorithm:intersection_test}).  We will use the
Conformal Separability Test to prove Theorem \ref{theorem:effective-implies-separable}
which states that an effective attack is always detectable with a finite sample validity guarantee.

This tool builds on conformal prediction \cite{book:vovkConformal,article:conformal-initial}.
At a high-level, conformal prediction provides a framework for rigorously, statistically analyzing machine
learning algorithms with probabilistic guarantees.  Researchers have noticed a strong link between conformal
prediction and anomaly detection \cite{phd:conformal-anomaly}.
Conformal prediction allows detecting deviations from
\emph{exchangeability}.  Exchangeability forces identically distributed marginals; hence, has been used
to detect distribution-shift in distribution-free settings.
Many of the applications of conformal prediction to anomaly detection consider the case where the sets of
predictions are singletons (so called confident predictions).  However, when
applying conformal prediction to neural networks, one often obtains sets with more than one element.  For example, \cite{paper:conformal-neural}
suggests that on some datasets, nearly $40\%$ of items have
non-singleton prediction sets.  We designed
the Conformal Separability Test to circumvent this multiple label problem; that is, to
enable poison detection even when the prediction is not ``confident'', thus making this more general than
anomaly detection.

We refer the reader to \cite{book:vovkConformal,book:conformalReliable} for details on conformal inference in general.
We sketch the central idea behind conformal prediction in Figure \ref{figure:conformal-is-about-crisp-sets}.
\begin{figure}
  \begin{center}
    \includegraphics[trim={5cm 2.7cm 15cm 1.5cm},clip,scale=0.4
    ]{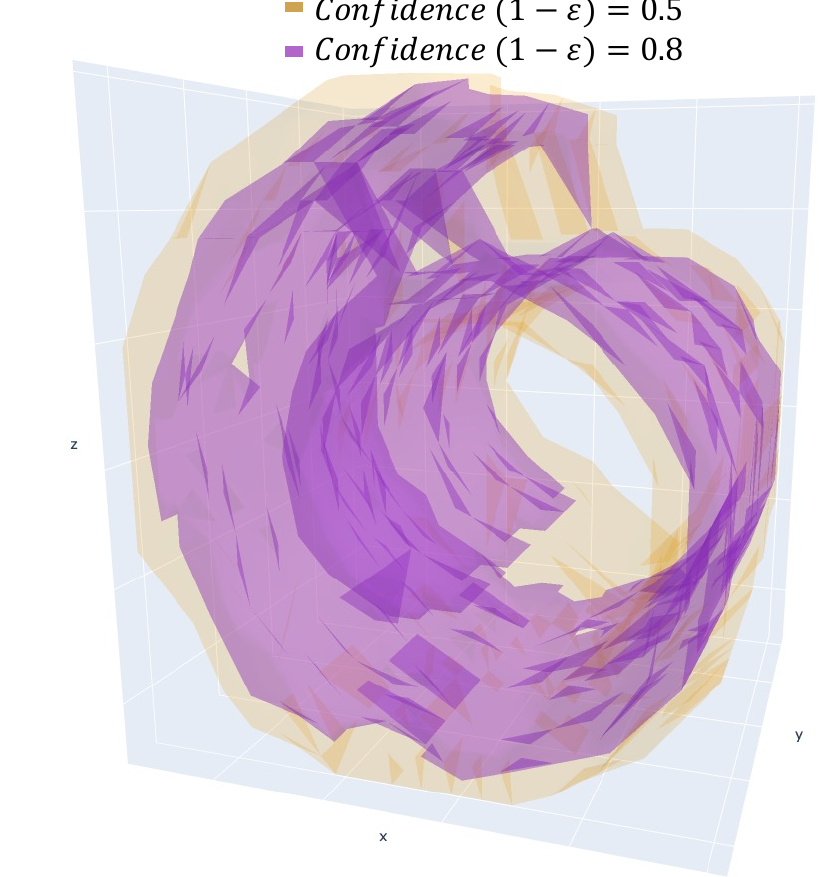}
  \end{center}
  \caption{An approximate visualization of conformal prediction sets
    at two confidence levels for the Swiss-roll dataset.
  }
  \label{figure:conformal-is-about-crisp-sets}
\end{figure}

Conformal prediction sets provide a set estimator in the distribution-free context that
only requires two mild assumptions: that there is a measurable test of \emph{non-conformity}
(a so-called \emph{non-conformity score}) and that our training
datasets are exchangeable.

\begin{definition}
  A non-conformity score is a family of measurable functions
  \[
    A:= \{ (\X\x\Y)^n \x \X \x \Y \to^{A_n} \real \, | \, n \in \mathbb{N} \}
  \]
  such that for each $A_n$ and for all $\sigma \in S_n$,
  \[
    A((b_1,\ldots,b_n),x,y) = A((b_{\sigma(1)},\ldots,b_{\sigma(n)}),x,y)
  \]
\end{definition}

We provide a few examples of non-conformity scores.

\begin{example}
  Suppose that $\Y$ is finite and that there is a $\Y$ dependent distance on $\X$; i.e. there is a metric $d_y(\blank,\blank)$ for each $y \in \Y$.
  Then for each $(b_1,\ldots,b_n)$, we obtain a non-conformity score by taking
  \[
    A_n((b_1,\ldots,b_n),x,y) = d_y(x,\mathsf{centroid}_y(b_1,\ldots,b_n))
  \]
  where $\mathsf{centroid}_y(b_1,\ldots,b_n)$ is the centroid of the
  set of points $(x_j,y)$ for $b_j=(x_j,y_j)$ with respect to the
  metric $d_y$.
\end{example}

The following example of a non-conformity score is often viewed as incredibly impractical.

\begin{example}
  Suppose $D$ is an exchangeable random variable in $\real^\mathsf{in\mbox{-}dim}\x\real^\mathsf{lab\mbox{-}dim}$.
  Write $M_D$ for the result of training a neural network on $D$
  with \emph{full backpropagation} on mean square error and a fixed learning rate $l$ and fixed weight initialization.  Then we obtain a
  non-conformity score:
  \[
    A(D,x,y) = d(M_D(x),y)
  \]
  where $d$ is the euclidean distance on $\real^\mathsf{out\mbox{-}dim}$.  This is symmetric in $d$ since we apply full
  backpropagation over the entire dataset at once.
\end{example}
Note that we can swap mean square error in the above with any symmetric loss function.

If one were to apply stochastic gradient descent or split $D$ into multiple batches, one would need a convergence
result to provide any sense in which $A$ is symmetric in the first $n$ arguments.  Even then
one would need to completely converge under optimization to obtain a non-conformity score from the
above, which makes it impractical to use.

To get around these issues and to simultaneously overcome the need to fully retrain a neural network on each application,
one often moves to a so-called inductive score, which we present as the following family of examples.
See for example, \cite{paper:conformal-neural,book:conformalReliable,book:vovkConformal} for concrete uses
of inductive conformal prediction on neural networks.

\begin{example}
  Suppose $A:= \{A_n : (\X\x\Y)^n \x \X\x\Y \to \real \, | \, n \in \mathbb{N}\}$ is a family of measurable
  functions (note we are \emph{not} assuming $A_n$ is symmetric in its first $n$ arguments).  Then for any $D_0 \in (\X\x\Y)^m$,
  we get a proper non-conformity score $A^{D_0}_n: (\X\x\Y)^n \x \X \x \Y \to \real$ by sending $D \in (\X\x\Y)^n$ to
  \[
    A^{D_0}_n(D,x,y) :=A_n(D_0,x,y)
  \]
  Note $A^{D_0}_n$ is trivially symmetric in $D$ since it discards it entirely.
\end{example}

The above is certainly non-trivial since using it has the same validity theorem as for starting with an ordinary
non-conformity score; however, one should keep in mind that they are essentially preconditioning the score on $D_0$
when stating conclusions.

A vital ingredient in performing conformal prediction is the p-value function.

\begin{definition}
  Let $A$ be a non-conformity score.  For each $n$ we obtain a
  measurable function $p_{n,A}: (\X\x\Y)^n \x \X\x\Y \to [0,1]$,
  called the \textbf{$p$-value function}, defined by setting $b:= (b_1,\ldots,b_n,(x,y))$ and then
  \[
    p_{n,A}((b_1,\ldots,b_n),x,y)
    := \frac{\left| \left\{ i \, | \, A_{n+1}(b,(x_i,y_i)) \geq A_{n+1}(b,(x,y)) \right\}\right|}{n+1}
  \]
\end{definition}

We call attention to the fact that $p_{n,A}$ is a $[0,1]$-valued random variable; whence, the preimage
of any generator is measurable.  In particular, \emph{conformal prediction sets} are defined
using preimages of the $p$-value function.
\begin{definition}\label{definition:conformal-pred-sets}
  Given a non-conformity score $A$, the \textbf{conformal prediction
  set} $\Gamma_{n,A}^{\epsilon}$ is defined as follows:
  \[
    \Gamma_{n,A}^\epsilon
    :=
    {p_{n,A}}^*((\epsilon,1])
  \]
\end{definition}

\begin{observation}\label{observation:uncurried-form-pred-set}
  In the literature on conformal prediction, conformal prediction sets are often
  described in \emph{uncurried form}.  In this form, we have a sample $D$ from the
  exchangeable distribution in question, and we write

  \[\Gamma_{n,A}^\epsilon(D) := \{(x,y) \, | \, p_{n,A}(D,x,y) > \epsilon\} \]

  When $D$ is a random variable then
  \[\Gamma_{n,A}^\epsilon \equiv \{(z_1,\ldots,z_{n+1}) \, | \, z_{n+1} \in \Gamma_{n,A}^\epsilon(z_1,\ldots,z_n)\}.\]
  We can also use DeFinetti notation for events and
  write $\{p_{n,A} > \epsilon\}$, $\{Z_{n+1} \in \Gamma\}$.  In this way, for random variables, $Z_1,\ldots,Z_{n+1}$, we can equate the events
  \[
    \{p_{n,A} > \epsilon\} \equiv \{Z_{n+1} \in \Gamma_{n,A}(Z_1,\ldots,Z_n)\} \equiv \Gamma_{n,A}^\epsilon
  \]

  In fact a further uncurrying is even used.

  \[
    \Gamma_{n,A}^\epsilon(D,x) := \{y \, | \, p_{n,A}(D,x,y) > \epsilon\}
  \]

  The final version requires more care when making probabilistic statements.  In
  particular we have no \emph{a priori} assumptions on the (marginal) distributions on the
  product space $\X\x\Y$ other than that $\Y$ is finite as used
  in, e.g., \cite{shafer2008tutorial,book:vovkConformal} for establishing validity when we
  uncurry all the way to the last $x$.
\end{observation}

Below, given a random variable $\<Z_1,\ldots,Z_{n+1}\>:\Omega\to (\X\x
\Y)^{n+1}$, we denote by $\mathbb{P}_{\langle Z_1,\ldots,Z_{n+1}\rangle}$
the pushforward measure on $(\X\x \Y)^{n+1}$.
\begin{theorem}\label{theorem:vovk-validity}[Vovk \emph{et al.} \cite{article:conformal-initial,article:conformal-transduction,book:vovkConformal}]
  For any exchangeable sequence of variables $\<Z_1,\ldots,Z_{n+1}\>: \Omega \to (\X \x \Y)^{n+1}$, we have
  \[
    \mathbb{P}_{\<Z_1,\ldots,Z_{n+1}\>}(\Gamma_{n,A}^\epsilon) \geq 1 - \epsilon
  \]
\end{theorem}
\begin{proof}
  See \cite{article:conformal-initial,conference:full-proof-conformal-plus-dist-shift} for a proof.
\end{proof}

Using the more common notation in the literature, described in Observation \ref{observation:uncurried-form-pred-set},
we would write $\mathbb{P}_{\<Z_1,\ldots,Z_{n+1}\>}(Z_{n+1} \in \Gamma_{n,A}^\epsilon(Z_1,\ldots,Z_n)) \geq 1-\epsilon$
instead of the form given in Theorem \ref{theorem:vovk-validity}.

The above uses the (exchangeable) distribution induced by the exchangeable random variable $\<Z_1,\ldots,Z_{n+1}\>$, but we
could equally well pull this straight back to the latent probability
space $\Omega$:

\begin{corollary}
  Given $\<Z_1,\ldots,Z_{n+1}\>: \Omega \to (\X\x\Y)^{n+1}$, an
  exchangeable sequence of random variables,
  $\mathbb{P}_{\Omega}(\<Z_1,\ldots,Z_{n+1}\>^*(\Gamma_{n,A}^\epsilon)) \geq 1-\epsilon$.
\end{corollary}

We depict the intuitive idea behind the Conformal Separability Test in Figure \ref{figure:conformal-sep-overlap}.
\begin{figure}
  \begin{center}
    \includegraphics[trim={7cm 3cm 7cm 2cm},clip,scale=0.35]{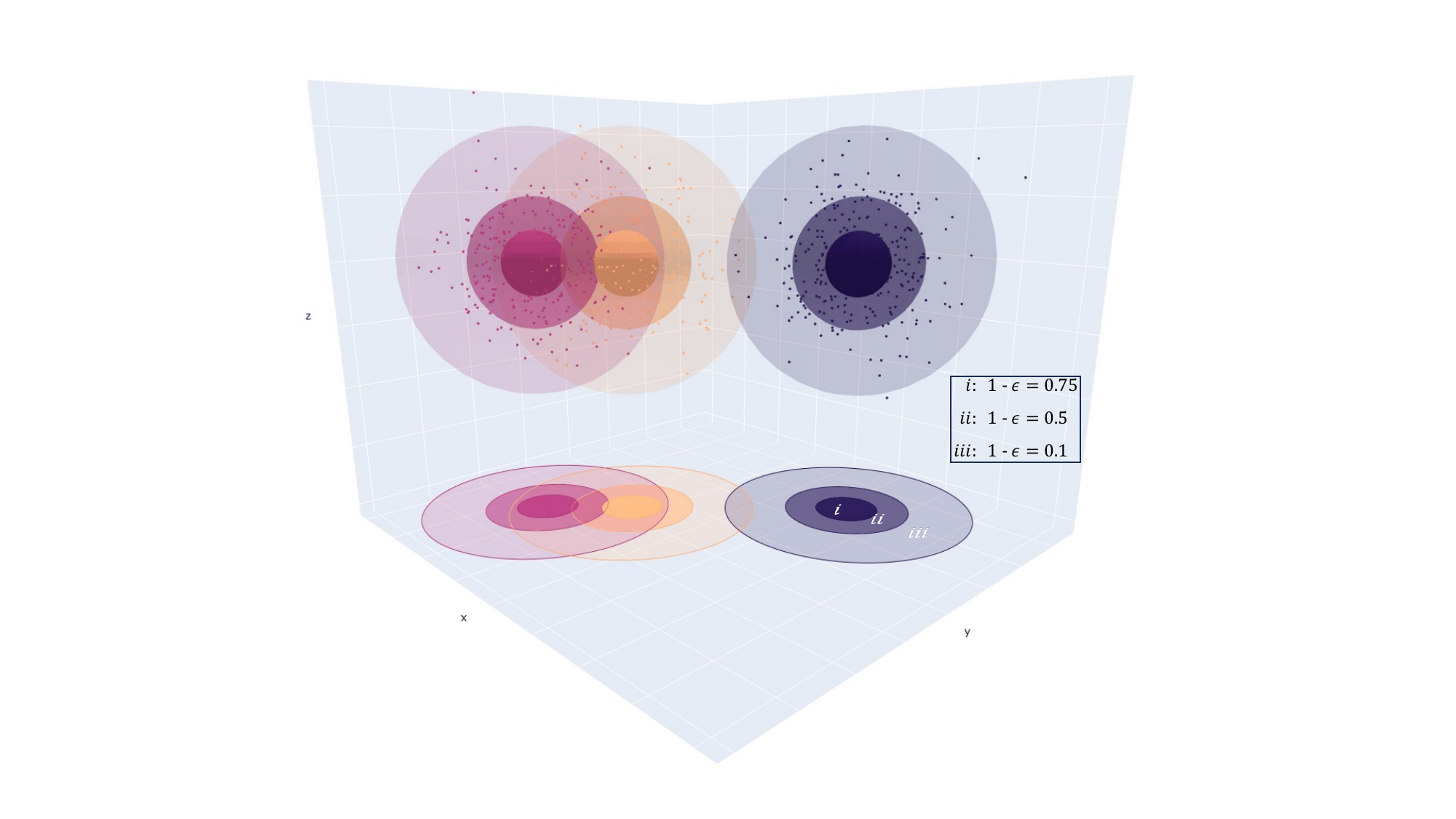}
  \end{center}
  \caption{Visualization of Conformal Separability. Conformal
    Separability is a conservative test for inequality of conformal
    prediction sets. Here the magenta
    and orange sets are not conformally separable until confidence
    is high.  The blue and orange spheres are much more separable, as
    they do not overlap
    until the confidence becomes remarkably low.  For the red and blue to overlap; one would need
    astronomically low confidence. This is akin to the test for normal
    distribution separation checking overlap of confidence intervals;
    but we are porting this to a distribution-free setting where the test is
  over arbitrarily irregular distributions on manifolds of high-dimension.}
  \label{figure:conformal-sep-overlap}
\end{figure}
There are less conservative tests than this, but they require making
additional trade-offs (e.g., fixing a model or accepting additional
computation). For example, $2$-sample $t$-tests (e.g., the Welch's test) can
detect significant difference with small p-values for largely overlapping confidence intervals, and one can do even
better when the variances are known to be the same or close.

In the following, we let $\mathbb{P}_{D_1,D_2}$, where $|D_1|=n$ and $|D_2|=m$, denote the induced (pushforward) measure
\[
  \Sigma_{\Z^n\x \Z^m} \to^{\<D_1,D_2\>^*} \Sigma_{\Omega}
  \to^{\mathbb{P}} [0,1].
\]
Let $\Delta: \Omega \to \Omega \x \Omega$ be $\omega \mapsto (\omega,\omega)$.  Then $\Delta^{*}(U\x V) = U\cap V$.
Then factor $\<D_1,D_2\>: \Omega \to \Z^n\x \Z^m$ as
\[
  \<D_1,D_2\> = \Omega \to^{\Delta} \Omega \x \Omega \to^{D_1 \x D_2} \Z^n\x\Z^m.
\]
Thus $\mathbb{P}_{D_1,D_2}(U \x V) = \mathbb{P}(D_1^*(U) \cap D_2^*(V))$ for every $U\in \Sigma_{\Z^n},V\in \Sigma_{\Z^m}$.
Thus $\mathbb{P}_{D_1,D_2}$ is canonically measuring an intersection event.  We may sometimes write simply $\mathbb{P}_{D_1,D_2}(U,V)$.

\begin{definition}
  Let a non-conformity score $A$ be given and assume $D_1$ and $D_2$
  are sequences of random variables. Fix $\epsilon,\gamma,\delta \in (0,1)$.
  $D_1$ and $D_2$ are \textbf{$A$-conformal
  $(\epsilon,\gamma,\delta)$-separable} when
  \[\mathbb{P}_{D_1,D_2}(\Gamma_{n_1,A}^{\epsilon} , \Gamma_{n_2,A}^{\gamma}) < \delta,\]
  where $\mathbb{P}_{D_1,D_2}$ is the probability distribution induced
  by $D_1$ and $D_2$.
\end{definition}
Since $\epsilon,\gamma,\delta >0$, when $D_1$ and $D_2$ have an $(\epsilon,\gamma,\delta)$-separator, they are are disequal on
a non-measure zero set, hence they are not almost surely equal.

Crucially, exchangeability forces an upper bound on separability.

\begin{theorem}\label{theorem:no-exch-strong-sep}
  \Copy{theorem:no-exch-strong-sep}{
    Let $A$ be a non-conformity score and
    suppose $D_1:\Omega \to \Z^n$ and $D_2: \Omega \to \Z^m$ are
    independent as two sequence-valued random variables such that each
    of $D_1$ and $D_2$ are individually exchangeable. Then $D_1$ and
    $D_2$ cannot be $A$-conformal
  $(\epsilon,\gamma,\delta)$-separable for any $\delta <(1-\epsilon)(1-\gamma)$.}
\end{theorem}
Note that Theorem \ref{theorem:no-exch-strong-sep} does \emph{not} assume
that $D_1$ and $D_2$ are individually IID; we are only assuming that,
individually, they are exchangeable, and that as random variables
valued in sequences these sequences are independent.
\begin{proof}
  The proof follows nearly immediately from independence and Theorem \ref{theorem:vovk-validity}:
  \[
    \mathbb{P}_{D_1,D_2}(\Gamma_{n,A}^\epsilon , \Gamma_{m,A}^\gamma)
    = \mathbb{P}_{D_1}(\Gamma_{n,A}^\epsilon) \cdot \mathbb{P}_{D_2}(\Gamma_{m,A}^\gamma)
    \geq (1-\epsilon) \cdot (1-\gamma).
  \]
\end{proof}

In Theorem \ref{theorem:no-exch-strong-sep}, we are measuring a set of the form $A\cap B$.
Also recall that for conformal prediction we view measuring $\Gamma_{n,A}^\epsilon$ intuitively
as $Z_{n+1}\in \Gamma_{n,A}^\epsilon$.  Strictly speaking, that statement is not well typed since as a set $Z_{n+1}\in \Gamma_{n,A}^\epsilon$ iff $p(D,Z_{n+1})>\epsilon$, hence the event
defined by the indicator of those predicates coincides.  To have a
similar intuition for separability, and indeed, to define the test, we need to ensure that the last random variable
of each sequences is the same.  The following is as such a consequence of Theorem \ref{theorem:no-exch-strong-sep}.
Thus if we have random variables of the form $\<D_1,Z_{\nu}\>$ and $\<D_2,Z_{\nu}\>$, we may soundly write
$Z_{\nu}\in \Gamma_{n+1,A}^\epsilon(D_1) \cap \Gamma_{m+1,A}^\gamma(D_2)$
for the pullback of $[0,\epsilon)\x [0,\gamma)$.

\begin{corollary}\label{theorem:no-exch-strong-sep-intuitive}
  Let $A$ be a non-conformity score and suppose $\<D_1,Z_{\nu}\>: \Omega \to \Z^{n+1}$ and $\<D_2,Z_{\nu}\>: \Omega \to \Z^{m+1}$
  are independent as two sequence valued random variables such that both $D_1,D_2$ are individually exchangeable.
  Then $\<D_1,Z_{\nu}\>$ and $\<D_2,Z_{\nu}\>$ cannot be $A$-confomal $(\epsilon,\gamma,\delta)$-separable for any $\delta < (1-\epsilon)\cdot(1-\gamma)$.
  Equivalently,
  \[
    \mathbb{P}_{\<D_1,Z_\nu\>,\<D_2,Z_\nu\>}(Z_{\nu} \in \Gamma_{n,A}^\epsilon(D_1) \cap \Gamma_{n,A}^\epsilon(D_2)) \geq (1-\epsilon)\cdot(1-\gamma)
  \]
\end{corollary}

Theorem \ref{theorem:no-exch-strong-sep} and Corollary \ref{theorem:no-exch-strong-sep-intuitive}
immediately determine an
algorithm,  our \emph{Conformal Separability Test},
for testing separability. Observe the probability measurement of $\Gamma$, for a sequence valued
random variable $D$, can be
reformulated as $\mathbb{P}_{D}(p_{n,A} > \epsilon) \geq 1-\epsilon$.  Note
that $p$ is a real-valued random variable and hence, by completeness
of the reals, we can take its \cadlag
completion $\overline{p_{n,A}}$ and observe that $\mathbb{P}_{D}(\overline{p_{n,A}} \geq \epsilon) \geq 1-\epsilon$.
For a brief introduction to the \cadlag completion see Exercises
5.13,5.14 in Chapter 1 and 1.17--1.21 in Chapter 2 of \cite{book:erhan-probability-theory}.

However, this can then be immediately reformulated in terms of the CDF of $\overline{p_{n,A}}$, denoted by $F_{p_{n,A}}$:
\[
  1 - F_{p_{n,A}}(\epsilon)
  = 1 - \mathbb{P}_{D}(p_{n,A} \leq \epsilon) = \mathbb{P}_{D}(\overline{p_{n,A}} \geq \epsilon) \geq 1-\epsilon.
\]
This implies that $F_{p_{n,A}}(\epsilon) \leq \epsilon$.

Observe that we have the following identity
\begin{align*}
  \mathbb{P}_{D_1,D_2}(\Gamma_{n,A}^{\epsilon},\Gamma_{m,A}^{\epsilon}) & = 1 -
  F_{\min(p_{n, A},p_{m,A})}(\epsilon).
\end{align*}
This allows us to state the correctness for Algorithm
\ref{algorithm:intersection_test} as follows.


\begin{algorithm}[tb]
  \caption{Conformal Separability Test}
  \label{algorithm:intersection_test}
  \textbf{Input}: Data samples $D_i=(x_1,y_1),\ldots, (x_{n_i},
  y_{n_i})$, for $i=1,2$ to be compared\\
  \textbf{Input}: Non-conformity score $A$\\
  \textbf{Input}: Test input $x$\\
  \textbf{Output}: $p_{\cap}(D_1,D_2,x)$
  \begin{algorithmic} 
    \STATE $p_{\cap}(D_1,D_2,x) := 0$
    \FOR{$y \in Y$}
    \FOR{$i\in [1,2]$}
    \STATE $p_{i,y} := p_A^{D_i}(x,y)$
    \ENDFOR
    \STATE $p_{\cap}(D_1,D_2,x) := \max(p_{\cap}(D_1,D_2,x),\min(p_{1,y},p_{2,y}))$
    \ENDFOR
    \STATE \textbf{return} $p_\cap(D_1,D_2,x)$
  \end{algorithmic}
\end{algorithm}

\begin{theorem}\label{theorem:correctness}
  Let $p_\cap(D_1,D_2,x)$ denote the output of Algorithm \ref{algorithm:intersection_test}.
  Then $\left|\left| p_\cap(D_1,D_2,x) - \mathbb{P}(\Gamma_{n,A}^{\epsilon},\Gamma_{m,A}^{\epsilon}) \right|\right|_\infty \to 0$
  almost surely.
\end{theorem}
\begin{proof}
  This is a special case of the Glivenko-Cantelli theorem \cite{article:glivenko-cdf,article:cantelli-cdf,article:definetti-cdf}, applied to
  Algorithm \ref{algorithm:intersection_test}, which computes $1-E$ where $E$ is the empirical
  CDF for $\min(p_{n_1,A},p_{n_2,A})$.
\end{proof}

A finite sample version of the above can be obtained from the
the Dvoretsky-Kiefer-Wolfowitz inequality \cite{article:finite-sample-cdf},
applied to Algorithm \ref{algorithm:intersection_test}, noting again
that it computes the complement of the empirical CDF (i.e. the
empirical version of $1-F_{\min(p_{n_1,A},p_{n_2,A})}$).  The finitary bound
incurs another constant and makes carrying it around harder.
We can work with the asymptotic version in this paper, since it is an asymptotic
result about convergence for a finite sample result.

\subsection{Empirical detection of poison attacks using conformal
separability}\label{subsection:detection-conformal}
We will now consider notions of effective attack.  We will show that
when given a dataset random variable $D$ and when $P$ denotes a trigger attack
applied to $D$, that $D,P$ are conformally separable.  In particular, we will
show that the event of observing an empty intersection between conformal prediction
sets is a rare event.  Moreover, an effective attack forces this intersection to be empty.
It then follows immediately that if on average the attack is successful,
then in expectation we will always see a rare event!

Before proceeding, we fix notation.  As before, we assume that $D: \Omega \to (\X \x \Y)^{n+1}$ is
an exchangeable random variable.  As common in statistics literature, we will identify
samples from the distribution induced by $D$ with the random variable $D$.  This is a convenient
abuse of notation since then the sample distribution of $D$ is precisely the distribution induced by $D$.
Likewise, if we post-compose
any $D_i$ with the trigger $t$ (in the sense of Def. \ref{definition:trigger_attack}), we induce a
distribution on $\X\x\Y$.  We will denote the
realizations of such samples by $(t_{x_i},t_{y_i})$.  We will denote by $P$, samples from the induced ``poisoned''
distribution on $(\X\x\Y)^{n+1}$.  In this section, we adopt the
abbreviation $p^{D}(x,y)$ for the p-value $p_{n,A}(D,(x,y))$.

\begin{definition}\label{definition:expiration}
  Given a nonconformity score $A$, dataset $D$ and $x \in \X$ we define the \textbf{expiry of $x$ for $D$} to be the
  random map $\tau^D(x)$ sending $y$ to  $p^D(x,y)$.  In our case,
  $\Y$ is finite and we interpret $\tau^D(x)$ as a random vector in $\real^{|\Y|}$.
\end{definition}

We can now give extra conditions that make detecting separation easy, and also relate
this to the test for poison.

\begin{example}
  Suppose the set $\Y$ consists of labels $a,b,c,d$, and $e$ with and the corresponding
  p-values for some $x$ are $0.76,0.05,0.89,0.09$, and $0.06$.
  Then we visualize $\tau^D(x)$ as follows:
  \begin{center}
    \includegraphics[scale=0.4]{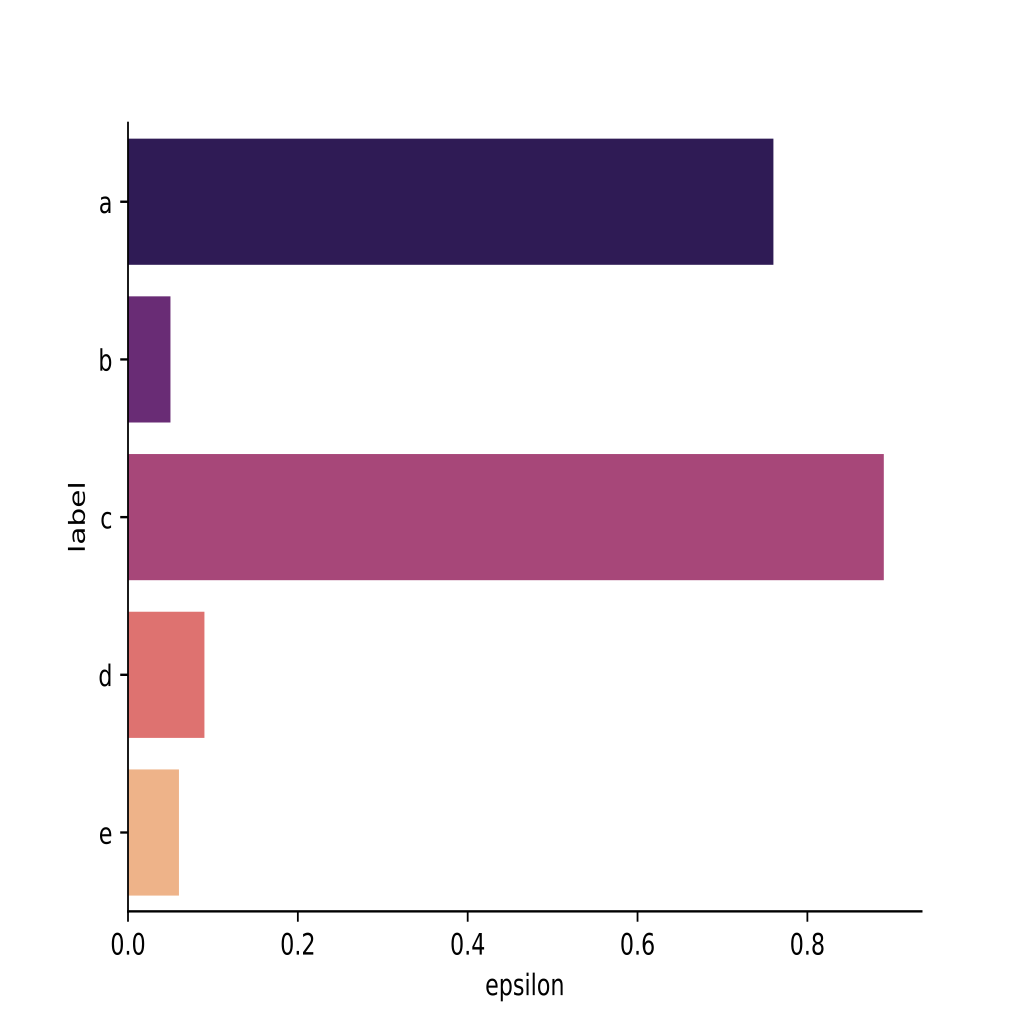}
  \end{center}
  This gives a picture of the conformal prediction set at $x$ and identifies exactly when
  a label leaves the set.
\end{example}

Expiry determines the separability of sets.  For example, suppose we are given $D_1,D_2,x$ with
expiry as indicated in the following figure:
\begin{center}
  \includegraphics[scale=0.4]{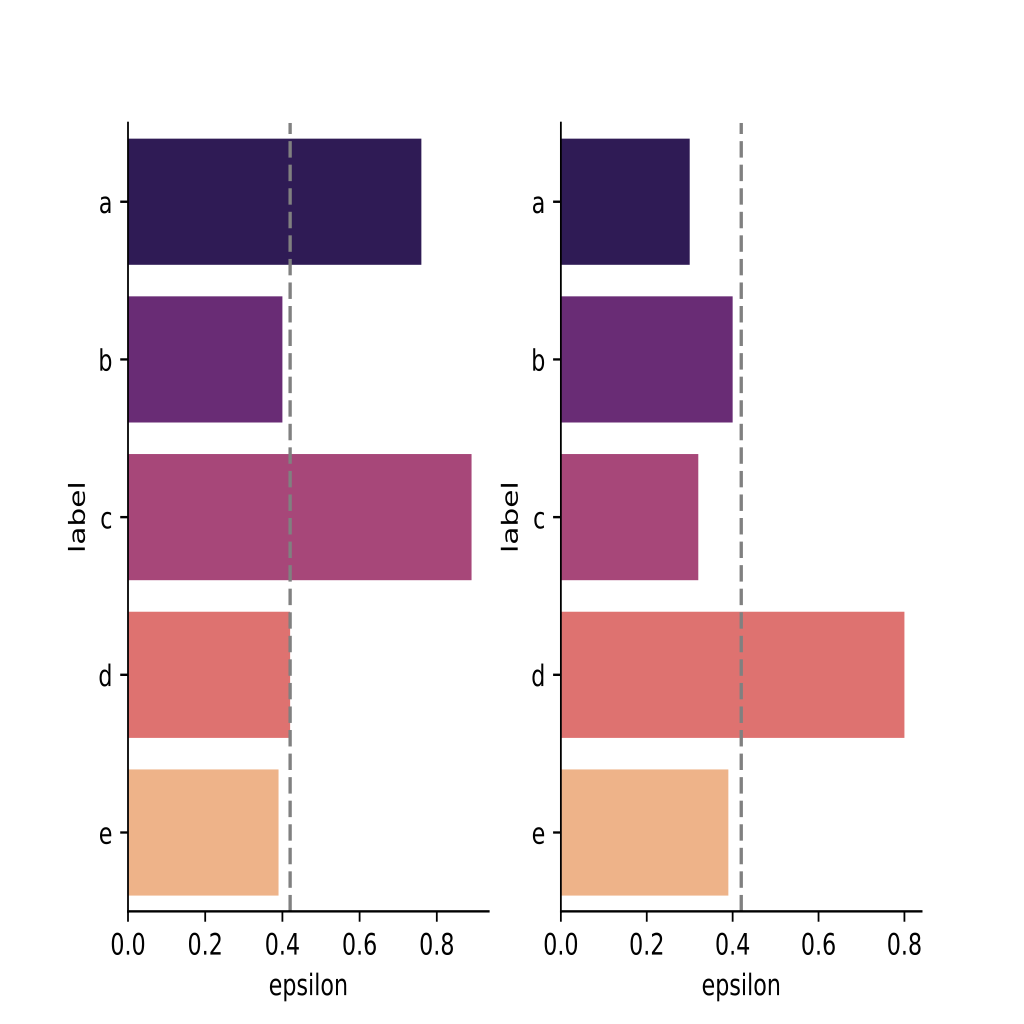}
\end{center}
The dashed line indicates the $\epsilon$ at which $x$ is a separator; the colors to the right of the
dotted line make it in the intersection at each $\epsilon$.


We will now describe how Algorithm \ref{algorithm:intersection_test}
can be used for the detection of poison attacks and we will then
evaluate the mathematical implications of Theorem
\ref{theorem:correctness} to this application. To this end, we will
first formalize the notion of what it means for an attack to be
effective.  As above, we assume that $A$ is a
fixed non-conformity score throughout this section and we caution that
many of the notions defined below depend on the choice of $A$.  We also
maintain the notation that $D$ is split as $D=D_C\uplus D_P$ and that $P:= D_C \uplus t(D_P)$.
In this section, any \emph{attack} referenced will be of the form given in Definition \ref{definition:trigger_attack}.

\begin{definition}\label{definition:non_trivial}
  A trigger attack is \textbf{non-trivial} at $(x,y)$ when
  $p^{D}(t_x,t_y)<p^{D}(t_x,y)$.
\end{definition}

By the strictness of the inequality in Definition \ref{definition:non_trivial}, we have that $y \ne t_y$ giving
an immediately obvious sense of non-triviality.  However, ultimately we want that $t_y$ will end up being
the strongest label when using the poisoned dataset $P$, and this notion of non-triviality also requires that
$t_y$ is not the strongest label to begin with.  In other words, the attack will have to change a ``belief'' in order to
be considered effective.  This brings us to define notions of effectiveness.

\begin{definition}\label{definition:weakly_effective}
  A trigger attack is \textbf{weakly effective} at $(x,y)$, when
  $p^{P}(t_x,y)<p^{P}(t_x,t_y)$.
\end{definition}
Note that \emph{weakly effective} at $(x,y)$ means that, in terms of expiry,
the label $t_y$ expires after $y$.  However, it does not indicate that the attack was confidently successful.
An attack is then intuitively effective when $t_y$ is the most confident label in $\tau^P$.  However, one should wonder if
this notion alone captures attack effectiveness.

\begin{example}
  The following expiry provides a sort of counterexample to bounding the confidence one may have in attack effectiveness
  given non-triviality and weak effectiveness alone:
  \begin{center}
    \includegraphics[scale=0.5]{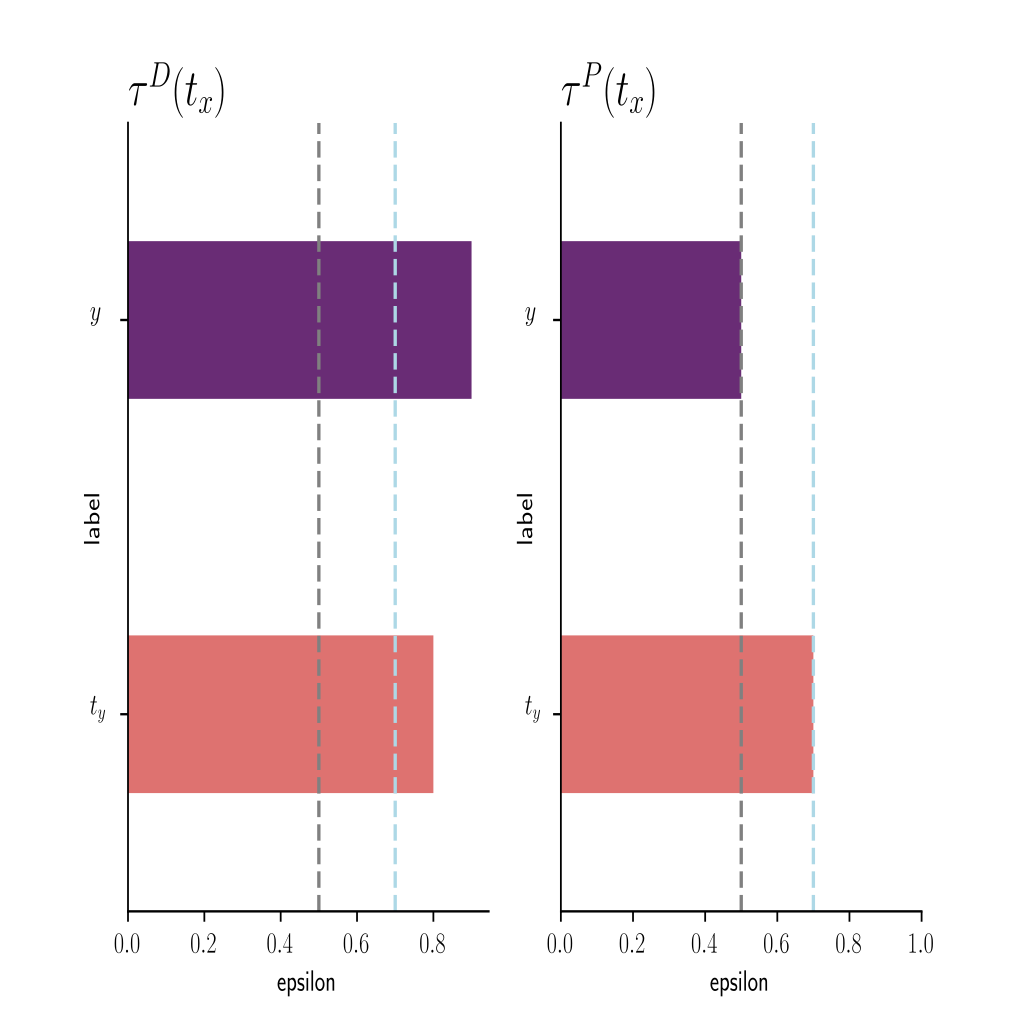}
  \end{center}
  The first problem is that for the entire region for which $t_y \in
  \Gamma^P$ and $y\notin\Gamma^P$ (between grey and lightblue),
  we have $t_y \in \Gamma^D$, hence not separable within that region.  Second, while the attack is formally effective, our confidence
  in $t_y$ decreased after the attack.  While this is not a problem for
  separation \emph{per se}, it is a
  problem for bounding our confidence in separation.
\end{example}

The above example indicates that, in addition to requiring $t_y$ to be the most confident label,
we need to take into account both $\tau^D,\tau^P$ simultaneously when creating
a meaningful notion of effective.  One possibility is to force relations like $p^P(t_x,t_y) > p^D(t_x,t_y)$.
We instead capture the effectiveness in a more general form, and to do so we make it quantitative.

\begin{definition}\label{definition:r-effective}
  An attack is \textbf{empirically $(1-r)$-effective} at $(x,y)$ given $D,P$ when
  \begin{enumerate}[{[$\mathsf{re}$.1]}]
    \item $p^D(t_x,t_y) < \min(r,p^D(t_x,y))$;
    \item $\max(\tau^P(t_x) - t_y ) \leq r < p^P(t_x,t_y)$,
  \end{enumerate}
  where we denote by $\tau^{P}(x)-y$ the vector obtained by
  projecting $\tau^{P}(x)$ onto the subspace corresponding to all but the $y$-th component.
\end{definition}

Intuitively, an attack is empirically $(1-r)$-effective when $\Gamma^P_\epsilon$
is at most the singleton $\{t_y\}$ and $t_y \not \in
\Gamma^D_\epsilon$ for $\epsilon > r$.

\begin{lemma}
  An empirically $(1-r)$-effective attack is both non-trivial and weakly effective.
\end{lemma}
\begin{proof}
  Suppose an attack is empirically $(1-r)$-effective.
  Since
  \begin{align*}
    p^D(t_x,t_y) < \min(r,p^D(t_x,y)) \leq p^D(t_x,y),
  \end{align*}
  the attack is non-trivial.
  Since $p^P(t_x,t_y) > r \geq \max(\tau^P(t_x) - t_y)$ we have that
  $p^P(t_x,t_y) > p^P(t_z,y)$ for any $y\ne t_y$.  Since the attack is non-trivial, $t_y\ne y$, hence
  the attack is weakly effective.
\end{proof}

The intuition behind a poisoned item being the most ``dominant feature'' can also be captured by requiring
that $p^P(t_x,t_y) > \max(\tau^D(t_x))$.  This kind of assumption,
which mirrors assumptions in \cite{khaddaj2023rethinking},
forces empirical $(1-r)$-effectiveness.

\begin{lemma}\label{lemma:attack-derived-effective-bound-typical}
  Suppose that an attack is non-trivial at $(x,y)$ and additionally satisfies
  \begin{align}
    \max(\tau^P(t_x)  - t_y) & < p^P(t_x,t_y)\\
    \max(\tau^D(t_x)) & < p^{P}(t_x,t_y).
  \end{align}
  Then the attack is empirically $(1-r)$-effective at $(x,y)$ for some $r$.
\end{lemma}
\begin{proof}
  Define
  \begin{align*}
    r & := \max\left(p^D(t_x,y),\max(\tau^P(t_x) - t_y )\right).
  \end{align*}

  To begin, [$\mathsf{re}.2$] is an immediate consequence of (1) and (2).
  Then note that $\min(\max(a,b),a)=a$ for any $a,b$.  Hence
  $\min(r,p^D(t_x,y)) = p^D(t_x,y)$.  Then [$\mathsf{re}.1$] follows from non-triviality.
\end{proof}

In fact the $r$ defined in the proof of Lemma
\ref{lemma:attack-derived-effective-bound-typical} typically much larger than needed.  Indeed, in practice $p^D(t_x,t_y)$ and
$\max(\tau^D(t_x)-t_y)$ are both very small, and quite comparable.  We can then push the $r$ much lower
depending on the relationship that these two quantities have.

\begin{observation}
  Suppose an attack is non-trivial at $(x,y)$.  Then
  \begin{enumerate}[(i)]
    \item If conditions (1),(2) of Lemma
      \ref{lemma:attack-derived-effective-bound-typical} hold, then
      the attack is always empirically $1-r$ effective for any $r$ with $\max(p^D(t_x,t_y),\max(\tau^P(t_x)-t_y)) < r < p^P(t_x,t_y)$;
    \item If Lemma \ref{lemma:attack-derived-effective-bound-typical}.(1) holds and additionally, $p^D(t_x,t_y) < \max(\tau^P(t_x)-t_y)$, then the attack is $1-r$ effective for $r:= \max(\tau^P(t_x)-t_y)$;
    \item If Lemma \ref{lemma:attack-derived-effective-bound-typical}.(2) holds and additionally, $\max(\tau^P(t_x)-t_y) < p^D(t_x,t_y)$ then
      the attack is empirically $1-r$ effective for any $r$ with $p^D(t_x,t_y) < r < p^P(t_x,t_y)$.
  \end{enumerate}
\end{observation}
\begin{proof}
  For (\emph{i}).  First note that from (1),(2) we have
  \[
    \max(p^D(t_x,t_y),\max(\tau^P(t_x)-t_y))\leq \max(\max(\tau^D(t_x)),\max(\tau^P(t_x)-t_y)) < p^P(t_x,t_y)
  \]
  so that the interval $\left(\max(p^D(t_x,t_y),\max(\tau^P(t_x)-t_y)),p^P(t_x,t_y) \right)$
  is non-empty and therefore such an $r$ exists.  [$\mathsf{re}.1$] then holds by non-triviality and construction of $r$ and [$\mathsf{re}.2$] holds by construction of $r$.

  For (\emph{ii}),[$\mathsf{re}.1$] holds by construction and non-triviality and [$\mathsf{re}.2$] follows immediately from (1).

  For (\emph{iii}), note that such $r$ exist by (2).  [$\mathsf{re}.1$] then holds by construction and non-triviality.  [$\mathsf{re}.2$]
  holds: $\max(\tau^P(t_x)-t_y) < p^D(t_x,t_y)<r<p^P(t_x,t_y)$ by assumption followed by the construction of $r$.
\end{proof}

Note that when an attack is $(1-r)$-effective we have the following.

\begin{lemma}\label{lemma:effective-forces-empty}
  Suppose an attack is empirically $(1-r)$-effective at $(x,y)$, given $D,P$.
  Then $\Gamma_{n}^{r}(D)(t_x) \cap \Gamma_{n}^{r}(P)(t_x) = \emptyset$.
\end{lemma}
\begin{proof}
  Suppose an attack is $(1-r)$-effective at $(x,y)$.
  From \textsf{[re.1]}, $t_y \not \in \Gamma_{n}^r(D)(t_x)$.
  From \textsf{[re.2]}, $\Gamma_{n}^{r}(P)(t_x)=\{t_y\}$, completing the proof.
\end{proof}

From the above we know that an empirically effective attack forces the
intersection of two prediction sets to be empty.  However, we have not
shown that we should be surprised by this.  The following lemma shows
that under an assumption of independence of samples, this is a rare situation.

\begin{lemma}\label{lemma:ind-exch-bounds-empty-meas}
  Suppose that $\<D_1,(X_{n+1},Y_{n+1})\>,\<D_2,(X_{n+1},Y_{n+1})\>$, as $(\X\x\Y)^{n+1}$-valued random variables
  are independent in $(\X\x\Y)^{n+1}$. Suppose also that both
  sequences are  exchangeable.  Then observing $p_\cap(D_1,D_2,x) \leq \epsilon$,
  (i.e. $\Gamma_{n}^{\epsilon}(D_1)(x) \cap \Gamma_{n}^{\epsilon}(D_2)(x) = \emptyset$)
  is $1-(1-\epsilon)^2$-rare.  The probability of $p_\cap(D_1,D_2,\blank) \leq \epsilon$ is less than
  $1-(1-\epsilon)^2$.
\end{lemma}
\begin{proof}
  From Corollary \ref{theorem:no-exch-strong-sep-intuitive}, we know

  \[\pr{(X_{n+1},Y_{n+1})\in \Gamma_{n,A}^\epsilon(D_1) \cap \Gamma_{n,A}^{\epsilon}(D_2)} \geq (1-\epsilon)^2.\]

  Thus,
  \begin{align*}
    & \pr{p_\cap \leq \epsilon}
    = \pr{\left( \Gamma_{n,A}^\epsilon(D_1), \Gamma_{n,A}^\epsilon(D_2) \right)^c}\\
    &= 1 - \pr{\Gamma_{n,A}^\epsilon(D_1), \Gamma_{n,A}^\epsilon(D_2)}
    \leq 1 - (1-\epsilon)^2
  \end{align*}

  From, the empirical approximation, Theorem \ref{theorem:correctness},
  the probability of observing $p_\cap \leq \epsilon$ is then asymptotically bounded
  by $1-(1-\epsilon)^2$.
\end{proof}

The main theoretical result of this paper summarizes the observations in this section
as follows.  Note that to make the statement precise we need to move from
the statement about effectiveness at a point to a probabilistic statement.
The most straightforward approach considers the ratio at a point, to a probabilistic statement,
and for this
we can take the expectation of the indicator for the set
of $(x,y)$ at which the attack is effective.

We summarize the above in the following theorem.
\begin{theorem}\label{theorem:effective-implies-separable}
  Suppose that $D,P$ are samples from $(\X\x\Y)^{n+1}$, that $D$
  is exchangeable, and that $P$ is an
  effectively poisoned dataset with respect to $D$.  Then $D,P$
  are conformally $(\epsilon,\epsilon,(1-\epsilon)^2)$-separable.
\end{theorem}

\section{Experimental Results}\label{section:experimental}

We cover experimental results in this section.  Overall our goal is to demonstrate three points.  First, that our
empirical false negative (FNR) and false positive rates (FPR) are competitive with the state of the art in proactive defense
provided by recent work \cite{usenix23-poison-detection}.
Second, that our false negative rate is better than the theoretical bound and that our false positive rate is low.
Third, to correct an \emph{en passant} claim made by \cite{geiping2020witches}.  In this paper they claim that so-called
label flipping and watermarking attacks (wherein the attacker only changes labels or superimposes a target onto images) can be detected by an application of ``supervision''.  They point to generic work that applies conformal
prediction by throwing away any inferences that are not conformal at
every layer of a neural network.  However, neither a
direct connection to detecting poison by non-conformity nor any
claim that there is a valid test for poison is made.  Second, they suggest that clean-label attacks using a bi-level optimization cannot be detected
by supervision.  Our experimental results can be used to show that one can indeed detect such ``weak attacks''
by conformal supervision.  Moreover, we show that clean-label
attacks can also be detected. In particular,
we demonstrate that the gradient-matching ``Witches' Brew'' attack can be detected using Conformal Separability.
We summarize results comparing to state-of-the-art in Table
\ref{table:experimental-high-level}.  A detailed analysis of
these individual results is given below.
\begin{table}
  \centering
  \begin{tabular}{cccc}
    \toprule
    Task $@$ & Missed Detections & False Alarms & Poison Success \\
    Poison Rate & (FNR)($\%$) & (FPR) ($\%$) & Rate ($\%$)\\
    \midrule
    Patch GTSRB (SOA)    & $0.2$ & $0.4$ & \multirow{2}{*}{$99.98$}\\
    Patch GTSRB (Ours)   & $1.4$ & $0.0$ \\
    \hdashline[1pt/1pt]
    Patch CIFAR (SOA)    & $1.7$ & $1.8$ & \multirow{2}{*}{$100$} \\
    Patch CIFAR (Ours)   & $1.2$ & $1.6$  \\
    \hdashline[1pt/1pt]
    Witches' Brew (Ours) & $9.0$ & $3.0$ & $35$\\
    \bottomrule
  \end{tabular}
  \caption{Our separability analysis provides competitive detection rates against
  other proactive defenses.  We also can detect the Witches' Brew attack well.}
  \label{table:experimental-high-level}
\end{table}

\paragraph{Experimental Setup}

In order to compare with prior literature, we are using the CIFAR10
and GTSRB datasets. We start by training a WideResNet $28\times10$
model, and a three layer simple CNN for each dataset respectively from
scratch and initialize the weights via He initialization
\cite{he2015delving}. We then create patches with size matching
the filter size of the first layer of the model architecture to
train it to kill that filter following the method discussed in
\cite{khaddaj2023rethinking}.
Random samples are then poisoned with the patch in each epoch based
on the pre-specified poison rates for 200 and 100 epochs.  This
process was carried out for both of the datasets. The position of the patch does not
necessarily have to be fixed.  Figure \ref{figure:sample-patch}
shows sample CIFAR10 and GTSRB images with the patches applied. Note that the images
are resized by about 7$\times$ for visualization purposes. However,
the model is trained with much smaller images ($32 \times 32$) and
patches ($7 \times 7$). For the optimization, we use the Adam optimizer with weight decays
and learning rate schedulers. We use the cross-entropy loss
function. 

\begin{figure}
  \centering
  \begin{subfigure}{0.3\linewidth}
    \includegraphics*[width=0.8\linewidth]{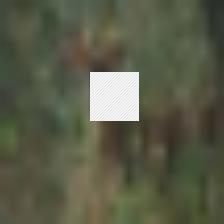}
  \end{subfigure}
  \begin{subfigure}{0.3\linewidth}
    \includegraphics*[width=0.8\linewidth]{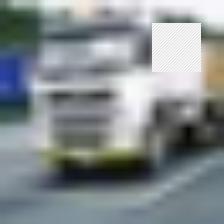}
  \end{subfigure}
  \begin{subfigure}{0.3\linewidth}
    \includegraphics*[width=0.8\linewidth]{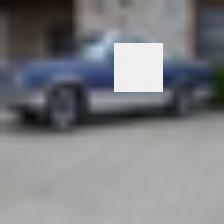}
  \end{subfigure}
  \begin{subfigure}{0.3\linewidth}
    \includegraphics*[width=0.8\linewidth]{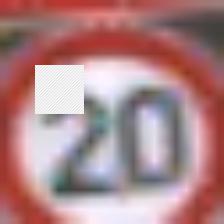}
  \end{subfigure}
  \begin{subfigure}{0.3\linewidth}
    \includegraphics*[width=0.8\linewidth]{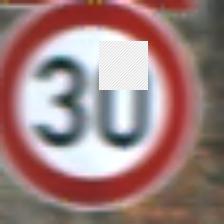}
  \end{subfigure}
  \begin{subfigure}{0.3\linewidth}
    \includegraphics*[width=0.8\linewidth]{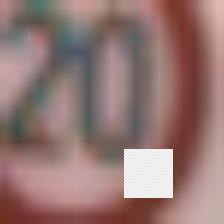}
  \end{subfigure}
  \caption{Sample CIFAR10 (top) and GTSRB (bottom) images with the patch applied.}
  \label{figure:sample-patch}
\end{figure}

\paragraph{Evaluation Metrics} Because we are interested in attacks
that satisfy the
\textit{effectiveness} requirement of our proposed framework,
we first look at \textit{attack success rates}.
The notion of success rate of an attack comes with some nuance. For the current paper,
we simply define the success rate at $S$ as the accuracy on a set of the form $\{(t_x,t_y) \, | \, (x,y) \in S \}$.
This, e.g., grossly over approximates the success rate since it does not take into account cases where $t_y = y$
nor the rate at which the clean model guesses $t_y$; however, it does seem to be the most common in the literature.

For evaluating our proposed Conformal Separability framework, we look
at both false negative and false positive rates. While we know by
Theorem \ref{theorem:effective-implies-separable} that we can
bound the false negative rate of detecting poison, it is worth investigating the false positive
rates. Additionally, \textit{attack optimality} is evaluated. The attacks we are interested in do not always push
towards optima on the poison rate spectrum.  I.e., such attacks need
not minimize imperceptibility (cf. Definition
\ref{definition:trigger_attack}).  This more relaxed position opens the door to a new
question: if poison is not administered at an optimal dose, can we detect ``enough'' poison
to render the attack inert?  Our analyses reveal an enticing point
on this front.  It seems that when poison is applied at either too
low or too high a rate the attack will not succeed.  This suggests
that there is a limited window of poison rates that yield effective attacks.  However, this introduces a hitherto unexplored
question: what is the maximum poison rate that we can make inert?  Of course, due to our validity guarantee,
we can always remove enough poison, but at what cost?  We will demonstrate bounds on the \textit{rates of poison} we
can deal with.

\subsection{Results on CIFAR10 and GTSRB }\label{subsection:direct_trigger}

We now consider our results in connection with the patch-based
poison attack summarized above on the CIFAR10 and GTSRB
datasets. Table \ref{table:poison-success-rate} displays the
relationship between poison rate, the
accuracy of the model on clean data and the success rate of the attack on poisoned data.
The first column provides the rate at which we altered data in the
training dataset. The second column shows that the attack does not alter validation
accuracy of the trained model on the canonical CIFAR10 or GTSRB test set.
The final column shows the success rate of the attack on CIFAR10 and GTSRB test set,
and demonstrates that attack rate has a breadth of successful deployment rates.
\begin{table}
  \centering
  \begin{tabular}{cccc}
    \toprule
    \textsc{Dataset}& \textsc{Poison Rate} (\%) & \textsc{Eval Clean} (\%) & \textsc{Success Rate} (\%)\\
    \midrule
    \rowcolor{orange!20}
    CIFAR10 & 0.0  & 93.35 & 10.0    \\
    CIFAR10 & 0.002 & 93.57 & 68.95   \\

    CIFAR10 & 0.01 & 93.35 & 97.57   \\
    CIFAR10 & 0.1 & 94.08 & 99.81   \\
    CIFAR10 & 0.2 & 94.01 &  100.00  \\
    CIFAR10 & 1 & 93.28 & 100.00   \\
    \hline
    \rowcolor{orange!20}
    GTSRB & 0.0  & 97.43 & 1.02  \\
    \rowcolor{orange!20}
    GTSRB & 0.002 & 97.65  & 1.54 \\
    GTSRB & 0.01 & 96.79  & 87.77 \\
    GTSRB & 0.1 &  97.78 & 99.57 \\
    GTSRB & 0.2 & 97.34  & 99.44 \\
    GTSRB & 1 & 97.92  & 99.98  \\
    \bottomrule
  \end{tabular}
  \caption{Test accuracy ($\%$) of a Wide ResNet $28 \times 10$ , and simple 3-layer CNN models trained from scratch respectively using clean and poisoned CIFAR10 and GTSRB
    data at multiple
    poison rates.  Eval clean shows that the model maintains the same performance on clean data.  Success rate
    shows that the attack becomes progressively more pronounced as the poison rate increases, taking a drastic leap
    in effectiveness from 0.002\% to 0.01\%.  The orange rows indicate
    the poison rates prior to the ``phase transition'' from ineffective
  to effective attacks.}
  \label{table:poison-success-rate}
\end{table}

To use Conformal Separability, we must identify a nonconformity score.  Previous work on
applying conformal prediction to neural networks (\cite{book:conformalReliable,book:vovkConformal})
used
\[
  A(D,x,y)  = -s_{x,y} \qquad \text{ or}\qquad
  A(D,x,y) = -\frac{\max_{j\ne y}s_{x, j}}{s_{x, y} +\gamma},
\]
where we write $s_{x,y} := \mathsf{softmax}(M_D(x))_y$ and $\gamma$ is a parameter used to adjust the sensitivity to the $y^\text{th}$ output \cite{paper:conformal-neural}.  These scores tended to produce slightly
less confident predictions, due to uniformly noisy, non-dominant outputs.  Also, the ``tuning'' parameter $\gamma$
becomes a sort of hyperparameter of the conformal prediction, which we want to avoid.
Yet such noise is captured well by
entropy.  Based on this observation, we consider more uniform distributions on the logits as
``less conformal'', and indeed such distributions
will have higher entropy. We then take as our
non-conformity score, the entropy on the logits, relative to the self-information on the output in question.
Explicitly, the non-conformity score we use is:
\[
  A(D,x,y) = -\frac{\mathsf{entropy}(\mathsf{softmax}(M_D(x)))}{s_{x,y} \cdot \ln(s_{x,y})}.
\]
Letting $T$ be the training set and $P=T_C\uplus t_{!}(T_P)$ be the poisoned set,
we apply Conformal Separability to $T$ and $P$ against a test set.


Recall that we derived a test for poison based on the Conformal
Separability Test, where we can confidently claim that a test item $x$
is poisoned when the measure of $\Gamma^{\epsilon}(T, x) \cap
\Gamma^{\epsilon}(T_r,x)$ where $T$ denotes a set of training data that is
known to be clean and $T_r$ denotes $T$ with the poison applied
at rate $r$.   For evaluation, we choose
hold-out sets of 500 poisoned and 500 clean items $(x,y)$ that were never seen in
training.  We then use the Conformal Separability Test, evaluated
over a range of poison rates $r$, to evaluate the intersections
$\Gamma^{\epsilon}(T,x)\cap\Gamma^{\epsilon}(T_r,x)$.  A
\emph{positive} identification of a purported poisoned item occurs
when this intersection is empty and a \emph{negative} when it is
non-empty.  Thus, a false positive occurs when a clean held-out
sample $x$ is identified as a positive and a false negative occurs
when a poisoned held-out sample $x$ is deemed a negative. The
resulting false
negative and positive rates at various conformal thresholds
$\epsilon$ and poison rates are recorded in Table
\ref{table:conformal-fnr-fpr-entropy}. As expected, false negative
rates decrease dramatically as soon as the poison rate has increased
to a sufficient level to yield an effective attack.  Interestingly,
in the case of CIFAR10, 9.8\% of the poisoned items are correctly
identified even when $r=0$ so that $T_r = T$.  I.e.,
$\Gamma^{\epsilon}(T,x)=\emptyset$ for 9.8\% of the
poisoned items $x$ from the CIFAR10 hold-out set.

\begin{table}
  \centering
  \begin{tabular}{cccccccc}
    \toprule
    & & \multicolumn{3}{c}{\textsc{False Negative}} &
    \multicolumn{3}{c}{\textsc{False Positive}} \\
    & & \multicolumn{3}{c}{\textsc{Rate} (\%)} & \multicolumn{3}{c}{\textsc{Rate} (\%)} \\
    \cmidrule(lr){3-5} \cmidrule(lr){6-8}
    & & \multicolumn{3}{c}{\small Conformal Threshold} &
    \multicolumn{3}{c}{\small
    Conformal Threshold}\\
    \textsc{Dataset} & \textsc{Poison rate} (\%) & 0.1  & 0.05 & 0.01 & 0.1 & 0.05 & 0.01 \\
    \midrule
    \rowcolor{orange!20}
    CIFAR10 & 0.0   & 90.2 & 90.2 & 90.2 & 0.0 & 0.0 & 0.0 \\
    CIFAR10 & 0.002 & 47.8 & 48.0 & 48.8 & 0.4 & 0.4 & 0.4 \\
    CIFAR10 & 0.01  & 6.2 & 6.4 & 6.6 & 0.0 & 0.0 & 0.0 \\
    CIFAR10 & 0.1   & 2.2 & 2.4 & 2.4 & 0.8 & 0.8 & 0.8 \\
    CIFAR10 & 0.2   & 1.2 & 1.4 & 1.4 & 1.0 & 1.0 & 1.0 \\
    CIFAR10 & 1.0   & 1.2 & 1.4 & 1.4 & 1.6 & 1.6 & 1.6 \\
    \hline
    \rowcolor{orange!20}
    GTSRB & 0.0   & 100 & 100 & 100 & 0.0 & 0.0 & 0.0 \\
    \rowcolor{orange!20}
    GTSRB & 0.002 & 99.5 & 99.7 & 100 & 0.16 & 0.16 & 0.0 \\
    GTSRB & 0.01  & 27.1 & 27.5 & 43.8 & 0.16 & 0.16 & 0.0 \\
    GTSRB & 0.1   & 2.2 & 2.3 & 14.1 & 0.0 & 0.0 & 0.0 \\
    GTSRB & 0.2   & 3.3 & 3.5 & 15.4 & 0.0 & 0.0 & 0.0 \\
    GTSRB & 1.0   & 1.4 & 1.7 & 14.3 & 0.0 & 0.0 & 0.0 \\
    \bottomrule
  \end{tabular}
  \caption{Evaluation of FNR and FPR on CIFAR10 and GTSRB where the
    evaluation is carried out on held-out poisoned and clean samples $x$
    (each of size 500). Following the Conformal Separability Test,
    the intersection
    $\Gamma^{\epsilon}(T,x)\cap\Gamma^{\epsilon}(T_r,x)$ is
    evaluated for each poison rate $r$ and conformal threshold
    $\epsilon$ considered, where $T$ is the respective training set
    and $T_r$ is the result of poisoning $T$ at rate $r$.
    Our FNRs for effective attacks come in lower than the theoretical bound established in
    Corollary \ref{theorem:no-exch-strong-sep-intuitive}.  Indeed, the sharp drop in
    FNR is correlated with
    the increase in success  rate from Table
    \ref{table:poison-success-rate}.  Even for unsuccessful attacks,
    we catch some of the poisoned samples. Additionally, our false positive rate shows that our test does not throw
    out good data.  The orange rows indicate
    the poison rates prior to the ``phase transition'' from ineffective
  to effective attacks.}
  \label{table:conformal-fnr-fpr-entropy}
\end{table}

\subsection{Witches-brew attack on CIFAR10}

One advantage of the Conformal Separability Test (Algorithm
\ref{algorithm:intersection_test}) is that
it opens the door to detection of attacks that might be too subtle in
their manipulation of the training data to be readily detected.  We
experimented with one such example,  the Witches Brew attack
\cite{geiping2020witches}, that poisons a dataset to
alter the classification of a fixed single target image $x_t$ by making
very small adjustments to a few (other) images.  That is, it can be
regarded as a kind of ``distributed'' version of an evasion attack in
which the distribution is shifted by small perturbation of several
images in order to yield misclassification of the target image.  In
general, the target image itself will not be manipulated in this
attack. The attack parameters are the poison rate, attack distance, and attack size.  The poison rate is as defined above.  Another parameter
of the the attack is the attack distance $\delta>0$ which provides a bound on the total amount of
change that each image can undergo (the attack stays with in a ball of
radius $\delta$ centered on the image being manipulated). The attack
size provides a budget for the number of pixels that can be
manipulated.  After setting these parameters, there are five ingredients
provided to the attack, as depicted in Figure \ref{fig:wb}.
\begin{figure}
  \centering
  \includegraphics*[scale=0.5]{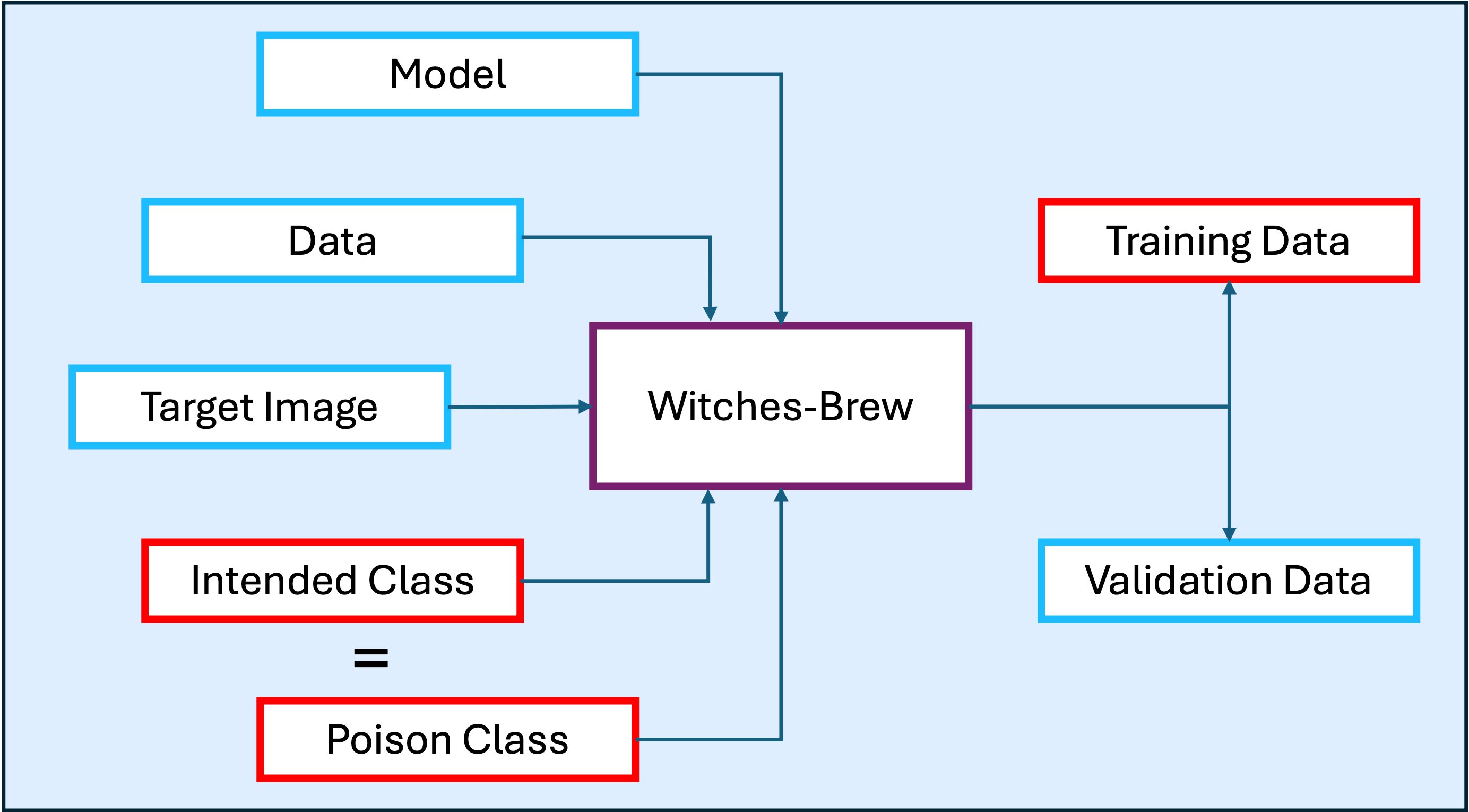}
  \caption{The Conformal Separability Test can be used to detect
    attacks that make very few changes to the training set.  We
    experimented with and were able to detect the Witches' Brew attack
    \cite{geiping2020witches} which pushes a target image to be
    misclassified by a model trained on the poisoned data, but where the
    target image and its label are not manipulated in the training set.
  Shown here are the inputs and outputs of the attack.}
  \label{fig:wb}
\end{figure}
The target image $x_t$ can be an image from any class.  The
\emph{intended class} is the class that we want the model to assign
to $x_t$.  For example, given
an image $x_t$ of a Labrador retriever and the intended class, ``sea lion'',
the attack aims to cause $x_t$ to be classified as ``sea lion.''  The
\emph{poison class} is the class from which images will be poisoned.
Note that the poison class need not be the actual class of the target
image. In our experiments we only considered the case where
the intended class and the poison class are identical, which was also
the setting of the experiments reported in \emph{ibid}.

Since the attack targets pre-selected images, we cannot simply apply the attack to additional test images after model training.
Thus, this is an expensive attack to carry out.  To analyze the success rate and detection capabilities of this attack,
we train clean and poisoned model pairs 100 times, and test on attacked images.
We fix a poison rate of $1.1\%$.  The attack success rate over the 100 tests is $35\%$.
With a conformal threshold of $0.15$, we obtain a false negative rate of $9\%$ thus
reducing the attack effectiveness to $3.5\%$.  The false positive rate is $3\%$.  Thus,
if we applied the mitigation to the data we would still have enough
data to adequately train. Thus, we have demonstrated that it is
possible to detect and mitigate even attacks that manipulate training data in
such a way as to impact performance on \emph{other} unmanipulated
data.  Such attacks are \emph{a priori} harder to detect than common
patch attacks.

\section{Conclusion and Future work}

In this paper we established that moving to a view of dataset poisoning as a stochastic problem captures
poisoning attacks well.  Indeed, we showed that the definition of a trigger attack can be formulated directly using
notions from distribution-free statistics (Definitions \ref{definition:non_trivial},\ref{definition:weakly_effective},\ref{definition:r-effective}).
We provided theoretical evidence (Theorem \ref{theorem:effective-implies-separable}) that these definitions
adequately capture the notion of poison.  Indeed for datasets $D_1,D_2$ which may be different and may have different
random errors in the models they produce, we showed that the statistical geometry of their conformal prediction sets
determines poison with provably bounded error on missed detection.  Moreover, we made no assumptions on the function
guiding the attack meaning that this bounded detection rate works, for example, even against adversaries that create
poisons using Turing-undecidable functions.  We also provided experimental evidence
that this test does not capture superfluous information that causes the geometry to capture too much more than the poisoned items
by showing experimentally low false positive rates (Section \ref{section:experimental}).  We also demonstrated experimentally
that even subtle ``clean-label'' attacks do admit a computationally effective test.  Finally, our theoretical and
experimental results provide a sound foundation for practical security against attacks in a scenario where the goal
is to preserve the integrity of a clean set of training data (the so-called proactive safety goal).

On the theoretical side, we showed that no poisoned item can be exchangeable with a clean distribution of data.  However,
we did not rule out that no poisoned item can be identically distributed.  Indeed one way to obtain this might be
to craft a poison which acts, for all intents and purposes as a singularity in the distribution, so that the surface of
poisoned items has measure zero.  For certain neural network models
and activation functions, it is not entirely
clear that this is possible. One way to achieve this would be to have a catastrophic surface that produces a cusp around
a poisoned item, which, upon projection in a future layer shows up as a proper singularity (to within floating point arithmetic).
Another possibility would be to construct clean label attacks that are identically distributed;
for example an attack like \cite{shafahi2018poison}.  One could imaging crafting such a scenario
that forces all items in a dataset to be in-distribution, but where the order and frequency are amiss (hence breaking a symmetry invariance).
It would certainly be interesting to know
conditions that permit identically distributed poisons.

On the theoretical and practical sides, it would be interesting to directly model or explore the surface of the
intersection of conformal prediction sets as a kind of statistical surface.  In the label classification example,
it would be interesting to know whether this surface must put a dynamics on the state space that pulls items absolutely
close together.  Given some assumptions on the distribution, for example, Lipschitz, it may be possible to then
further our result and derive a provable guarantee for detection methods for a class of poisons which rely on distance
(e.g. a k-nearest-neighbors test) in some latent space.  This could also then provide a way to bootstrap the test with even
less clean data, and for some classes of poison, possibly with no clean data required to boot.

We showed that having access to some clean data allows us to develop a practical test for poison and is immediately
applicable to proactively protecting datasets.  What we do not currently know is any kind of lower bound on the
amount of clean data required.  It would be theoretically interesting if there is some class of poison attacks
which provably require \emph{some} clean data to detect; it would also be interesting to relate the amount of clean data required
to some kind of complexity measure of the model.

One curiosity we discovered in performing our experiments is that often the relationship between poison-rate
and poison effectiveness is sharp.  For example, for both the GTSRB and CIFAR datasets, there seems to be almost
a phase transition in how a small amount of increase in poison-rate affects the success rate (see Table \ref{table:poison-success-rate}) at around a $0.002$ poison rate.
To our knowledge, there has been no analysis of phase transitions in the study of dataset poisoning, and it would be interesting to know
when there can be no constant bounding the relationship of poison-rate to success-rate.

\printbibliography

\end{document}